\documentclass{fundam}

\usepackage{amsmath,amsfonts,amssymb, mathtools, stmaryrd}

\usepackage{tikz}
\usepackage{cmll}
\usetikzlibrary{shapes.geometric}
\usetikzlibrary{fit}
\usepackage[all]{xy}
\usepackage{booktabs}

\newcommand{\cS}{\mathcal{S}}
\newcommand{\cD}{\mathcal{D}}
\newcommand{\cB}{\mathcal{B}}
\newcommand{\cR}{\mathcal{R}}
\newcommand{\cT}{\mathcal{T}}
\newcommand{\bA}{\mathbb{A}}
\newcommand{\CF}{\mathcal{CF}}

\newcommand{\Var}{\mathit{Var}}

\newcommand{\sort}[1]{\text{\tt #1}}
\newcommand{\Const}[1]{\text{\tt #1}}
\newcommand{\Rel}[1]{\text{\tt #1}}

\newcommand{\FactSet}{\sort{FSet}}
\newcommand{\NeFactSet}{\sort{NeFSet}}
\newcommand{\Pattern}{\sort{Pat}}
\newcommand{\TPattern}{\sort{Pat}^t}
\newcommand{\STPattern}{\sort{Pat}^{\text{\it st}}}
\newcommand{\PPattern}{\sort{Pat}^p}
\newcommand{\TPPattern}{\sort{Pat}^{\text{\it tp}}}
\newcommand{\Query}{\sort{Qy}}
\newcommand{\DQuery}{\sort{DQy}}
\newcommand{\Nat}{\sort{Nat}}
\newcommand{\Bool}{\sort{Bool}}
\newcommand{\Fact}{\sort{Fact}}
\newcommand{\Cond}{\sort{Cnd}}
\newcommand{\Empty}{\emptyset}

\newcommand{\True}{\top}
\newcommand{\False}{\bot}
\newcommand{\ffalse}{\text{\bf f}}
\newcommand{\ttrue}{\text{\bf t}}

\newcommand{\closed}{\mathit{cl?}}
\newcommand{\Dml}{\sort{Dml}}
\newcommand{\From}{\nabla}
\newcommand{\Eval}{\mathit{eval}}
\newcommand{\Tr}{\mathit{tr}}
\newcommand{\Trb}{\overline{\mathit{tr}}}
\newcommand{\CFd}{\CF_{\Sigma,\cD}}

\newcommand{\QLang}{\mathcal{Q}_{\Sigma, \cD}}
\newcommand{\Pos}{\mathit{Pos}}
\newcommand{\dml}{\mathbf{dml}}
\newcommand{\query}{\mathbf{qry}}
\newcommand{\cond}{\mathbf{cnd}}

\newcommand{\Del}[1]{[#1]_0}
\newcommand{\Keep}[1]{[#1]_!}
\newcommand{\Ret}[1]{[#1]_?}
\newcommand{\Rett}[1]{[#1]_{?}}
\newcommand{\Fresh}[1]{[#1]_{\mathbf{n}}}
\newcommand{\FrMark}{\mathbf{n}}
\newcommand{\dummy}{\surd}
\newcommand{\Con}{\rhd}

\newcommand{\Nom}{\mathit{nom}}

\newcommand{\StateC}{\sort{State}^c}
\newcommand{\StateQ}{\sort{State}^q}
\newcommand{\StateD}{\sort{State}^d}

\newcommand{\StackC}{\sort{Stk}^c}
\newcommand{\StackQ}{\sort{Stk}^q}
\newcommand{\StackD}{\sort{Stk}^d}

\newcommand{\NodeC}{\sort{Frm}^c}
\newcommand{\NodeQ}{\sort{Frm}^q}
\newcommand{\NodeD}{\sort{Frm}^d}

\newcommand{\Res}{\mathfrak{r}}
\newcommand{\Ans}{\mathfrak{a}}
\newcommand{\Init}{\mathrm{I}}
\newcommand{\Not}{{\neg}}
\newcommand{\New}{\mathfrak{n}}
\newcommand{\Fail}{\mathfrak{f}}
\newcommand{\Sat}{\mathfrak{s}}
\newcommand{\Ok}{\mathit{\dummy}}
\newcommand{\Token}{\sharp}

\newcommand{\Unfl}{\text{unf}}
\newcommand{\Fld}{\text{fld}}
\newcommand{\init}{\text{init}}

\newcommand{\LFalse}{\lambda_\False}
\newcommand{\LBool}{{\lambda_{\{\_\}}}}
\newcommand{\LNotUnfl}{{\lambda_{\neg}^\Unfl}}
\newcommand{\LNotFld}{{\lambda_{\neg}^\Fld}}
\newcommand{\DisjUnfl}{{\lambda_{\vee}^\Unfl}}
\newcommand{\DisjFldf}{{\lambda_{\vee;\ffalse}^\Fld}}
\newcommand{\DisjFldt}{{\lambda_{\vee;\ttrue}^\Fld}}
\newcommand{\QuantInit}{{\lambda_{\exists}^\init}}
\newcommand{\QuantUnfl}{{\lambda_{\exists}^\Unfl}}
\newcommand{\QuantFldt}{{\lambda_{\exists;\ttrue}^\Fld}}
\newcommand{\QuantFldf}{{\lambda_{\exists;\ffalse}^\Fld}}
\newcommand{\QuantEnd}{{\lambda_{\exists}^{\text{end}}}}
\newcommand{\LSat}{\lambda_{\text{sat}}}
\newcommand{\LFact}{\lambda_{\text{fact}}}
\newcommand{\LEmpty}{\lambda_{\Empty}}
\newcommand{\LUnUnfl}{\lambda_{\Con}^\Unfl}
\newcommand{\LUnFldOk}{\lambda_{\Con;\Ok}^\Fld}
\newcommand{\LUnFldEmpty}{\lambda_{\Con;\Empty}^\Fld}
\newcommand{\LCondUnfl}{\lambda_{\text{cond}}^\Unfl}
\newcommand{\LCondFldf}{\lambda_{\text{cond};\ffalse}^\Unfl}
\newcommand{\LCondFldt}{\lambda_{\text{cond};\ttrue}^\Unfl}

\newcommand{\FromInit}{{\lambda_{\From}^\init}}
\newcommand{\FromUnfl}{{\lambda_{\From}^\Unfl}}
\newcommand{\FromEnd}{{\lambda_{\From}^{\text{end}}}}

\newcommand{\LAns}{\lambda_{\text{ans}}}

\newcommand{\Collapse}{\lambda_{\text{col}}}
\newcommand{\LDummy}{\lambda_{\dummy}}

\newcommand{\DmlNew}{\lambda^{\text{new}}}
\newcommand{\DmlFail}{\lambda^{\text{fail}}}

\newcommand{\Succ}{\sort{Yes?}}
\newcommand{\yes}{\text{\it yes}}

   \usepackage[colorlinks = true, linkcolor = black, urlcolor = black, citecolor = black, anchorcolor = black]{hyperref}

\begin{document}

\setcounter{page}{141}
\publyear{2021}
\papernumber{2095}
\volume{184}
\issue{2}

   %%\finalVersionForARXIV
    \finalVersionForIOS

\title{A Non-Deterministic Multiset Query Language}

\author{Bartosz Zieli\'nski\thanks{Address  for correspondence: Department of Computer Science, Faculty of Physics
                   and Applied Informatics, University of \L{}\'od\'z,  Pomorska 149/153 90-236 \L{}\'od\'z, Poland. \newline \newline
          \vspace*{-6mm}{\scriptsize{Received November 2021; \ revised November 2021.}}}
\\
Department of Computer Science\\
 Faculty of Physics and Applied Informatics, University of \L{}\'od\'z\\
Pomorska 149/153 90-236 \L{}\'od\'z, Poland\\
bartosz.zielinski@fis.uni.lodz.pl}

\maketitle

\runninghead{B. Zieli\'nski}{A Non-deterministic Multiset Query Language}

\begin{abstract}
  We develop a  multiset query and update language executable in a term rewriting system.
  Its most remarkable feature, besides non-standard approach to quantification
  and introduction of fresh values, is
  non-determinism --- a query result is not uniquely determined by the database.
  We argue that this feature is very useful, e.g.,  in modelling user choices during
simulation or reachability analysis of a data-centric business process --- the intended application of our work. Query evaluation is implemented by converting
  the query into a terminating term rewriting system and normalizing the initial term which encapsulates the current database. A normal form encapsulates a query result.
  We prove that our language
  can express any relational algebra query.
  Finally, we present
  a simple business process specification framework (and an example specification).
  Both syntax and semantics of our query language is implemented in Maude.
\end{abstract}


\section{Introduction}

In a data-centric
approach to business process modelling (see, e.g.,
 \cite{hull2008artifact,calvanese2013foundations}), specification of
  data transformation during case execution
 is an integral part of the business process model.
This new paradigm requires new tools and formalisms for effective
specification, simulation and validation.
Task-centric models are commonly formalized using
Petri nets (see, e.g., \cite{van1998application,van2005yawl}).
Adapting Petri net-based formalizations to data-centric models is,
however, problematic: While
   simple transformations on data
 {\em can} be represented directly  within a coloured Petri net, Petri nets lack facilities for complex data processing and querying.
 Even so, there exists a large amound of literature (see e.g., \cite{rosa2011decidability, lasota2016decidability, montali2016model}) devoted to enriching Petri nets with with data processing capabilities and automated verification of their properties.
 Recent paper \cite{montali2017db}
introduced {\em DB-Nets} --- an attempt to integrate coloured Petri nets with relational databases.
The use of two separate formalisms
complicates verification and simulation (though it corresponds to actual  implementations
of BPM systems). In the following paper \cite{montali2019db} a subset of database operations was implemented inside coloured Petri nets with name creation and transition priorities.

\medskip
Conditional term rewriting \cite{meseguer1992conditional}
 was proposed  as an alternative (if less popular) generic framework for specification of dynamic systems \cite{RewLogSem}. It subsumes a variety of Petri nets \cite{stehr2001rewriting}
and their simulation is one of popular applications of the term rewriting system Maude \cite{padberg2016model,kheldoun2017formal}).
More precisely, it is well known (see e.g., \cite{stehr2001rewriting}) that coloured Petri nets can be implemented by multiset rewriting systems, and, conversely, rewriting systems which rewrite multisets of terms representing colour tokens can be interpreted as Petri nets: just identify places with colour tokens, and each rewriting rule of the form
\begin{equation*}
a_1a_2\ldots a_n\Rightarrow b_1b_2\ldots b_m
\end{equation*}
with a transition with input arcs from $a_1$, $a_2$, $\ldots$, $a_n$ and output arcs to $b_1$, $b_2$, $\ldots$, $b_m$. Other constructs, such as inhibitor arcs can be easily implemented with conditional rewriting rules. Rewriting systems are more general than Petri Nets since they are not limited to
rewriting multisets (on the other hand, Petri nets are much better supported by tools and programming libraries). However,  rewriting systems
still share with Petri nets the limitation and inconvenience of not directly
supporting bulk, complex operations on data which involve some kinds of quantification. For example, suppose that the state of a data driven business process related to e-commerce is represented by a multiset of terms. In particular, terms of the form $\Rel{item}(p, c)$ denote the presence of a product $p$ in the basket of a customer $c$. Suppose now that $c$ cancels the case, and so we need to remove all $c$'s items from the multiset. Describing removal of a single (nondeterministically chosen) item is easy
with the rule $\Rel{item}(c, x)\Rightarrow\emptyset$ where $x$ is a variable, but specifiying in a rewriting system (or a Petri net) that all of them need to be removed before the next business step is  more complex (though clearly possible) and would require auxilliary tokens and conditional rewrite rules
corresponding to inhibitor arcs in a Petri net preventing other transitions as long as there are still some $c$'s items present in the rewritten multiset. Thus, a high-level query language adding quantifier
constructs on top of conventional rewriting system or a Petri net is clearly desirable, particularly if rewriting systems or Petri nets are to become convenient formalisms to model data driven business processes.

\medskip
In this paper we present
 an expressive multiset query and update language $\QLang$, designed to be
 executable in a term rewriting system, useful for a unified
 and somewhat ``Petri nettish''  formalization of data-centric business processes. The connection with Petri nets is admittedly tenuous and follows from the fact that the language acts on data represented as multisets of terms which could be viewed as tokens (see the discussion above).
Since   this language specifies changes to data, instead of being a reimplementation
of relational calculus, it contains linear-like features fitting a term rewriting implementation.
Most remarkably, $\QLang$ is non-deterministic --- the result of a query or update
is not, in general, uniquely determined by the database. This permits
modelling user choices, just like in the case of Petri nets.
 $\QLang$ consists of three
 sublanguages, parametrized with respect to a signature $\Sigma$ and $\Sigma$-algebra of facts  $\cD$:
\begin{enumerate}
\item Language $\QLang^\cond$ of conditions
(Boolean queries)
which can be used independently as constraints, or as components of
queries in $\QLang^\query$ and $\QLang^\dml$.
\item Data manipulation language$\QLang^\dml$. A DML query $Q$ in $\QLang^\dml$
defines new facts to be added to the database.
Some of the old facts  used in constructing the new ones may be deleted.
\item Language $\QLang^\query$ of queries which only return facts but do not change the database.
Both syntax, and to some extent semantics of  $\QLang^\query$
is a restriction of syntax and semantics of  $\QLang^\dml$.
\end{enumerate}

A query $Q$ in $\QLang^\alpha$, $\alpha\in\{\cond, \query,\dml\}$, is given semantics by assignment of a rewriting system $\cR_{\Sigma, \cD}^\alpha(Q)$.
To evaluate $Q$ in a database $F$
we start with an initial term $\Init_Q(F)$.
A normal form of $\Init_Q(F)$
wraps a result of $Q$'s evaluation: a Boolean value indicating validity of a condition,  a query
answer or a new database resulting from execution of a DML query.
As remarked above, while $\cR_{\Sigma, \cD}^\alpha(Q)$ is always terminating
(i.e., there are no infinite execution paths, so we always do get {\em some}
result from evaluating $Q$), it is not confluent (i.e., divergent execution paths may not eventually converge) in general, hence we may get distinct results depending on nondeterministic choices. We identify syntactic constraints on queries in each of the sublanguages which ensure confluence for the rewriting system, and hence determinism for the results of the query.
 $\QLang$ shares with the language introduced in \cite{Bartek2017}
 a non-standard approach to variable binding. The approach
avoids  problems with capture-free substitutions without dispensing with explicit variables,
but at the price of non-compositionality: The surrounding context may determine whether
 a variable in a subterm is free or bound {\em in this subterm}.
$\QLang^\dml$ supports introduction of fresh values to the database,
which is used, e.g., to generate identifiers for newly created artifacts or to simulate user input (cf.~\cite{montali2016model}).

Since queries in $\QLang$ are converted to a term rewriting system,
verification of a
business process specified as a set of DML queries
(see Section~\ref{ExampleActionsSection}) can be assisted with
symbolic reachability analysis techniques based on narrowing (see e.g., \cite{meseguer2007symbolic}) .
 We plan to expand on this idea in future research.
Note that our results in this article regarding the confluence
 of the rewriting systems to which a particular subclass of
 queries compiles to, may be relevant for narrowing
(see e.g., \cite{fay1979first,hullot1980canonical}, c.f. \cite{alpuente2009termination}). E.g., narrowing a confluent system may provide a more efficient search procedure than in the case of a non-confluent one.

\subsection{Prior work}

\label{SectPrior}

The present paper builds on the previous paper \cite{Bartek2017}
(cf.~\cite{bpBDAS2017}) where a multiset query language
executed in Maude was proposed. $\QLang$
shares many similarities with the language described in \cite{Bartek2017}, particularly the treatment of quantification. It has, however, distinct syntax (with, e.g., fact markings in quantifiers) and distinct semantics. We consider both languages  to be alternatives, each of which with its own strengths.
A side-by-side comparison is presented in Table~\ref{Comparison}.
Observe that while the language described in \cite{Bartek2017} can be defined both in the set and multiset setting (in the former case we match multisets of facts modulo idempotence in addition to commutativity, associativity and identity), here we assume exclusively multiset setting. This is because the language described here is implemented through multiset rewriting. Making multiset constructor idempotent would make it impossible to consistently replace matched terms, a feature which is crucial for our formalism to behave sensibly.
E.g.,  given a rule $a\Rightarrow b$, term $ab=_{\mathcal{A}}aab$ rewrites in one step both to $bb=_{\mathcal{A}}b$ (intended) and $abb=_{\mathcal{A}}ab$ (not intended).

\begin{table}[h]
\vspace*{-2mm}
\footnotesize
\begin{center}
\caption{Comparison between $\QLang$ and $\CFd(X)$ from \cite{Bartek2017}}
\label{Comparison}\vspace*{-1mm}
\begin{tabular}{@{}p{0.46\textwidth}p{0.46\textwidth}@{}}
\toprule
\multicolumn{1}{c}{$\QLang$} & \multicolumn{1}{c}{$\CFd(X)$}\\
\midrule
Semantics of (possibly) non-deterministic queries is based on translation into
term rewriting systems.
&
Semantics of queries
 is based on term matching on the metalevel. Queries are always deterministic.\\
 \addlinespace[3pt]
 Queries return facts in the same signature as the database against which the query is evaluated. &
 Queries return objects of an arbitrary signature (as long as it contains a ``union'' operator).\\
  \addlinespace[3pt]
 DML queries construct multisets of facts to be added to the current database.
 Some of the current facts used in the construction may be deleted.
 &  DML expressions construct a pair of sets or multisets of facts --- those to be deleted from and those to be added to the current database.
 \\
  \addlinespace[3pt]
Fresh values are introduced through ``virtual fresh facts''.
Actual freshness in ensured through rewrite rules which define the semantics of DML query.&
Attributes of input facts can be marked by special sorts as fresh.
Testing framework ensures that values injected in those columns are actually fresh. \\
It is assumed that the database is a {\em multiset} of facts &
The language can be defined both in set and multiset setting\\
 \addlinespace[3pt]
  \multicolumn{2}{c}{Treatment of variable binding  is identical in both languages}\\
  \bottomrule
\end{tabular}
\end{center}
\end{table}

$\QLang$, has some similarities to matching logic \cite{rosu2010matching,rocsu2017matching}, as both are based on term matching. They
have different purposes, however, and different syntax and semantics. Unlike matching logic
statements,  $\QLang$ queries are non-deterministic.
Matching logic is used for software verification, while
$\QLang$ is a query language intended to be a {\em component} of the system.
Finally, matching logic has conventional quantifiers, whereas we use a  non-standard  quantification over ``relation patterns''.

CINNI \cite{cinni} is a generic calculus of substitutions implemented in Maude
which combines de Bruijn indices with explicit names to solve the problem
of capture-free substitutions.
To avoid the associated complexity,
we decided not to use conventional variable binding implemented, e.g., using CINNI.

Our  quantification
over ``relation patterns'' instead of variables, resembles quantifier constructs
 in description logic
\footnote{We are grateful to Prof. Andrzej Tarlecki for this observation.}
  \cite{baader2003description}. Our syntactic construct, however, uses explicit variable names and is not limited to binary relations where only the second column is bound by
 the quantifier.

  Data-centric business process models are formalized in a variety of ways.
First-order logic and its
 restricted variations (see e.g.,
\cite{de2012verification,hariri2011foundations,calvanese2015description,abdulla2016recency}),
datalog \cite{chen2006automating}, and  UML \cite{merouani2014formalizing},
are popular choices.
Those formalisms are excellent for the specification of data models, and they come with expressive query
languages; they are, however, not so ideally suited for modelling change, because of frame problems
\cite{mccarthy1969some}, where it is not always obvious what information is modified and what stays the
same. Rewriting formalisms \cite{meseguer1992conditional}, which are explicit about scope of change have
a clear advantage here. On the other hand, a great deal of work has been devoted to formal verification
of logic based business process formalism, see e.g., \cite{li2017verifas,deutsch2019verification}
in the context of hierarchical artifact systems. For another example see \cite{calvanese2019formal} where
a data aware extension of BPMN was proposed together with SMT (satisfiability modulo theories)  based
verification techniques.

In \cite{seco2018reseda} a language called Reseda,  was introduced for specification of data driven
business process. The language integrates data description with behaviour. What makes it relevant as a
prior work to the present paper is that Reseda's semantics is defined by associating with a Reseda
program a transition system. Such a program can then be executed by rewriting data in accordance with the
transition rules, similarly to the execution of the language described here.

As we remarked earlier,
a recent paper \cite{montali2017db}
introduced {\em DB-Nets} which integrate coloured Petri nets with relational databases.
Since the use of two separate formalisms
complicates verification and simulation in the following paper \cite{montali2019db} a subset of database
operations was implemented inside coloured Petri nets with name creation and transition priorities. Thus,
the motivation of \cite{montali2019db} is analogous to the motivation  of this paper, but in the world of
Petri nets instead of term rewriting systems. There are however two important differences: First, our
language is meant to provide complete specification of data driven business processes, whereas in
\cite{montali2019db} the business process is still specified as a Petri net, and there is just an
interface between relational queries (perhaps implemented inside the net itself) and the main net
describing the process. Secondly (and this is what makes the first point possible) we do not simply
implement a conventional relational dml and query language in a rewriting system. Instead, we implement a
{\em linear} and {\em non-deterministic} query language which can emulate user choices and creation of
new objects.

\subsection{Preliminaries on term rewriting}

We recall basic notions related to
 term rewriting \cite{meseguer1992conditional,RewLogSem},
 and many sorted equational logic \cite{Mmemb_1998}.

 Let $S$ be a poset (partially ordered set). A family $X=\{X_s\;|\;s\in S\}$ of sets
is called {\em an
 $S$-sorted set} if $X_s\subseteq X_{s'}$ whenever  $s\leq s'$.
 We abbreviate
$x\in\bigcup X$ as $x\in X$. We write $x:s$ iff $x\in X_s$.

\medskip
 {\em An algebraic signature} $\Sigma=(\Sigma_S, \Sigma_F)$
 consists of a finite poset of sorts
 $\Sigma_S$ and a finite set $\Sigma_F$ of function symbols.
The set of function symbols $\Sigma_F$ is $\Sigma_{S}^+$-sorted, where $\Sigma_{S}^+$
is the set of finite, non-empty sequences of elements of $\Sigma_{S}$ partially ordered
with $s_0\cdots s_{n}\leq t_0\cdots t_{m}$ iff $m=n$, $s_0\leq t_0$ and $s_i\geq t_i$ for all
$i\in\{1,\ldots,n\}$.
Traditionally we write $f:s_1\ldots s_{n}\rightarrow s_0$ when
$f\in(\Sigma_F)_{s_0\ldots s_n}$, where we denote by $(\Sigma_F)_{s_0\ldots s_n}$ the set of function symbols of sort $s_0\ldots s_n$. This explains the somewhat confusing ordering on $\Sigma_{S}^+$:
we are covariant on return value and contravariant on arguments.
 Symbols $c:\rightarrow s$
 are called {\em constants} of sort $s$.
 A $\Sigma$-algebra $\bA$ is an assignment of
a set $\llbracket s\rrbracket_\bA$
 to each  $s\in\Sigma_{S}$ such that $\llbracket s\rrbracket_\bA\subseteq \llbracket s'\rrbracket_\bA$
 if $s\leq s'$,
and a function
 $\llbracket f\rrbracket_\bA : \llbracket s_1\rrbracket_\bA \times\cdots\times
 \llbracket s_{n}\rrbracket_\bA\rightarrow\llbracket s\rrbracket_\bA$
 to each  $f:s_1\ldots s_n \rightarrow s$ in $\Sigma_F$.
 Let $V:=\{V_s\;|\;s\in\Sigma_{S}\}$ be a $\Sigma_{S}$-sorted set of variables.
 A term algebra $\cT_\Sigma(V)$ has ``sort-safe'' terms as elements
 and function symbols interpreted by themselves. We denote by $\cT_\Sigma$
 the algebra of ground $\Sigma$-terms.
 We often use mixfix
 syntax where underscores in the function name correspond to consecutive arguments.
 Thus, if $\Sigma_F$ contains  $\_+\_:A\;A\rightarrow A$ and $0:\rightarrow A$ then
  $0+0$ is a ground term of sort $A$.
Positions in a term are denoted by strings of positive integers.
Denote by $\varepsilon$ the empty string, and by $t|_\kappa$
the subterm of $t\in\cT_\Sigma(V)$ at position $\kappa\in\mathbb{Z}^*_+$
(if defined), i.e.,
$t|_\varepsilon := t$, and  $f(t_1,\ldots,t_n)|_{k\kappa}:=t_k|_\kappa$.
Let $\Pos(t):=\{\kappa\in\mathbb{Z}^*_+\;|\;t|_\kappa\;\text{is defined}\}$.
If $\kappa\in\Pos(t)$ and $u$ is a term of the same sort as $t|_\kappa$,
 then we denote by $t[u]_\kappa$ the result of replacing  $t|_\kappa$ in $t$ with $u$.
We use a standard notation for substitutions. Let $\vec{a}=a_1,\ldots,a_n$ be a list of terms,
 $\vec{v}=v_1,\ldots,v_n$ a list of distinct variables. Then we denote
  $\sigma=\{\vec{a}/\vec{v}\}=\{a_1/v_1,\ldots,a_n/v_n\}$
when $\sigma(v_i)=a_i$, $i\in\{1,\ldots,n\}$, and $\sigma(v)=v$ for any variable
$v\notin\{v_1,\ldots,v_n\}$.

A $\Sigma$-algebra may be defined  as a quotient of $\cT_\Sigma$ by a congruence generated by a
set $A\cup E$ of equalities, where equalities in $A$, referred to as {\em equational attributes}, define structural properties such as associativity, commutativity, or identity, and $E$ consists of
 conditional equalities interpreted as directed simplification rules
 on $\cT_\Sigma$.
 It is assumed that simplifications terminate and are confluent, hence each $t$ has the unique (modulo $A$)
irreducible form $t{\downarrow_{E/A}}\in\cT_\Sigma$ representing a class of $t$ in
$\cT_\Sigma/{=}_{A\cup E}$.

 Simplification with respect to equalities computes values. The behaviour  is represented with rewritings.
A rewriting system $\cR=(\Sigma, A, E, R)$ consists of a  signature $\Sigma$,
a set of equations $A\cup E$ were $E$ defines  confluent and terminating (modulo $A$) simplifications on
$\cT_\Sigma$ , and a finite set $R$ of conditional rewriting rules of
the form $\lambda:t_1\Rightarrow t_2\ \mathit{if}\ C$, where optional condition $C$ is a conjunction
of equalities, and $\lambda$ is the rule's label.
A one-step rewrite $u\xrightarrow{\lambda}_\cR u'$ from $u$ to $u'$ using such a rule
is possible if
there exists a position $\kappa$, term $v$, and a substitution $\sigma$ such that
$u=_Av$,
 $v|_\kappa=_A\sigma(t_1)$, $u'=_A v[\sigma(t_2)]_\kappa$ and
$\sigma(C)$ is satisfied. We write $u\rightarrow_\cR u'$ iff there exist terms
$s,s'\in\cT_\Sigma$ and  a label $\lambda$ of a rule in $R$ such that
$u{\downarrow}_{E/A}=_A s$, $s\xrightarrow{\lambda}_\cR s'$, and
$s'{\downarrow}_{E/A}=_Au'{\downarrow}_{E/A}$.
We denote by  $\rightarrow_\cR^{+}$ and $\rightarrow_\cR^{*}$
the transitive and reflexive-transitive closures of $\rightarrow_\cR$.
We also write $u\rightarrow_\cR^!u'$ if $u\rightarrow_\cR^{*}u'$ and there is no $u''$
such that $u'\rightarrow_\cR u''$.
If $\cR$ is implied by the context, we omit $\cR$ from arrows.

Variants of the following definition and easy to prove lemma appear in the literature (see, e.g.,$\,$\cite{huet1980confluent}):
\begin{definition}
\label{SemiConflMod}
Let $(X, \rightarrow)$, where $\rightarrow\subseteq X\times X$, be a transition system.
Assume that $\equiv$ is an equivalence  on $X$ which is a bisimulation on
$(X,\rightarrow)$. We call $\rightarrow$ {\em semiconfluent at $x\in X$ modulo $\equiv$}
if for all $y,y'\in X$ such that $y\leftarrow x\rightarrow^*y'$ there exist
$z,z'\in X$ such that $y\rightarrow^*z\equiv z'\;{}^*\!\!\leftarrow y'$.
We call $\rightarrow$ {\em semiconfluent modulo $\equiv$}
if $\rightarrow$ is semiconfluent at all $x\in X$.
We call $\rightarrow$ {\em confluent at $x\in X$ modulo $\equiv$}
if for all $y,y'\in X$ such that $y\;{}^*\!\!\leftarrow x\rightarrow^*y'$ there exist
$z,z'\in X$ such that $y\rightarrow^*z\equiv z'\;{}^*\!\!\leftarrow y'$.
We call $\rightarrow$ {\em confluent modulo $\equiv$}
if $\rightarrow$ is confluent at all $x\in X$.
\end{definition}
\begin{lemma}
Semiconfluence modulo equivalence implies confluence modulo equivalence.
\end{lemma}

\section{Multisets of facts, fresh facts and patterns}
\label{MofFandPSect}

Our queries are evaluated against, or  act on, finite multisets of facts.
Duplicate facts can be genuinely useful and removing them is computationally expensive.
If necessary, duplicates can be removed explicitly or, better, one can ensure
that no duplicates are introduced in the first place by judicious choice of DML operations.
In fact, SQL is a multiset query language as well, hence by using multisets we
are closer to the actual relational database practice than  formal systems based on sets.

$\QLang$ is parametrized with respect to a signature of facts $\Sigma$ and a  $\Sigma$-algebra
of facts  $\cD$. $\Sigma_S$ must contain sorts $\Fact$ and $\Bool$ for facts and Booleans, respectively.
All constructors for facts
are contained in $\Sigma_F$. Facts are reifications of predicate instances.
A typical
fact has the form $f(a_1,\ldots, a_n)$, where $f:s_1\ldots s_n\rightarrow\Fact$
is a fact constructor.
 $\cD$ defines all the data types used in facts and is
  specifiable in terms of directed equations and equational attributes.
  $\cD$ must define  Boolean connectives and Boolean-valued
 equality $\_=\_:s\;s\rightarrow\Bool$ for all $s\in\Sigma_{S,K}$.
We assume that all ground terms of sort $\Bool$ simplify to either $\ttrue$ or $\ffalse$.
This is non-trivial: Define a function $f:\Nat\rightarrow\Bool$
with a single equation $f(0)=\ttrue$. Then $f(1)$ is fully reduced and
distinct from both $\ttrue$ and $\ffalse$.

\smallskip
{\bf Multisets of facts.}
The signature of  multisets extends $\Sigma_S$ with sorts ($\sort{Ne}$)$\FactSet$
 of finite (non-empty) multisets of facts. The subsort ordering is given by
$\Fact<\NeFactSet<\FactSet$. In particular, each fact is a non-empty multiset of facts.
Finite multisets of facts are constructed with an associative and commutative binary operator
$
\_\circ\_:\FactSet\;\FactSet\rightarrow\FactSet
$ (cf.~\cite{thielscher1998introduction}) with identity element $\Empty:\rightarrow\FactSet$. Operator
 $\_\circ\_$ is subsort overloaded with the additional declaration
 $\_\circ\_:\FactSet\;\NeFactSet\rightarrow\NeFactSet$.
Thus, a multiset constructed from a multiset and a non-empty multiset is non-empty.

\smallskip
{\bf Freshness and nominal sorts.}
Support for creation of fresh values is a common requirement
(cf.~\cite{montali2016model}): Identifiers for new objects
 must not belong to the present nor any past  active domain of the database. To understand why reusing identifiers from past domains is bad consider situation where a new business object is created with the identifier of a previously deleted one. In this case the attempt to verify if the deleted object is present in the final database may yield an incorrect affirmative answer.
We support creation of fresh values of {\em nominal sorts} only.
Usually this suffices, and
freshness  for non-nominal data types is problematic
(cf.~\cite{ochremiak2014nominal}). A sort $s$ is nominal (relative to $\Sigma$-algebra $\cD$)
 if values of this sort
have no non-trivial algebraic or relational structure beside equality. In particular,
for nominal $s$, $s'\leq s\leq s''$ if and only if $s'=s=s''$. To create values of each nominal sort $s$ we have constructor
$\imath^s_{\_}:\Nat \rightarrow s$ which belongs {\em neither}  to $\Sigma$ {\em nor} to the
signature of $\QLang$.

\begin{example}
\label{Basket1}
Consider a client basket database.
Identifiers of customers, products and baskets have sorts
$\mathfrak{c}$, $\mathfrak{p}$, and $\mathfrak{b}$, respectively.
We use two fact constructors:
$\_\Rel{owns}\_:\mathfrak{c}\;\mathfrak{b}\rightarrow\Fact$
and $\_\Rel{in}\_:\mathfrak{p}\;\mathfrak{b}\rightarrow\Fact$.
Multiset of facts
$(\imath^{\mathfrak{c}}_1\;\Rel{owns}\; \imath^{\mathfrak{b}}_1)\circ
 (\imath^{\mathfrak{p}}_2\;\Rel{in}\;\imath^{\mathfrak{b}}_1)\circ
 (\imath^{\mathfrak{p}}_3\;\Rel{in}\;\imath^{\mathfrak{b}}_1)$
denotes the state in which customer
$\imath^{\mathfrak{c}}_1$ is the owner of basket  $\imath^{\mathfrak{b}}_1$ containing products
$\imath^{\mathfrak{p}}_2$ and $\imath^{\mathfrak{p}}_3$. Using multisets instead of sets can be useful:
multiset $(\imath^{\mathfrak{p}}_3\;\Rel{in}\;\imath^{\mathfrak{b}}_1)\circ
(\imath^{\mathfrak{p}}_3\;\Rel{in}\;\imath^{\mathfrak{b}}_1)$
denotes  the situation where basket $\imath^{\mathfrak{b}}_1$ contains two items of
$\imath^{\mathfrak{p}}_3$.
\end{example}

Constructing values of nominal sorts from natural numbers simplifies creation of fresh values: to ensure freshness one can construct the new value with the smallest natural number which was not used so far.
To keep track of those ``smallest unused naturals'', and
to make retrieval of  fresh values similar to retrieval of data we use
fresh facts of sort $\Fact^\FrMark$ unrelated to $\Fact$.
 For each nominal sort $s$ we have a single-argument constructor of the form
$C_s:s\rightarrow\Fact^\FrMark$ which wraps value $\imath^s_n$ such that for all $m\geq n$, $\imath^s_m$
was never used before.  When  fresh value of sort $s$ is requested, we
return $\imath^s_n$ and update the fresh fact to $C_s(\imath^s_{n+1})$.
Fresh facts are combined into (non-empty) multisets of sort ($\sort{Ne}$)$\FactSet^\FrMark$
using commutative and associative operator
$\_\circ\_:\FactSet^\FrMark\;\FactSet^\FrMark\rightarrow\FactSet^\FrMark$ with
identity $\Empty:\rightarrow\FactSet^\FrMark$.
To facilitate bulk updates of fresh facts needed in the semantics of DML queries, we define
the following function:
\begin{equation}
\label{FreshBulkUpd}
\upsilon:\FactSet^\FrMark\rightarrow\FactSet^\FrMark,\quad
\upsilon\bigl(C_{s_1}(\imath^{s_1}_{m_1})\circ\cdots\circ C_{s_n}(\imath^{s_n}_{m_n})\bigr)
=C_{s_1}(\imath^{s_1}_{m_1+1})\circ\cdots\circ
C_{s_n}(\imath^{s_n}_{m_n+1}).
\end{equation}

{\bf Patterns.}
Quantifiers in $\QLang$ quantify over patterns (of sort $\Pattern$) containing
 non-ground multisets of facts and fresh facts marked
by modalities, which control retention of matched facts (i.e., whether upon matching they are removed temporarily, permanently, or not at all from the database),
 and syntactically wrap fresh facts (i.e., fresh facts can appear in the pattern only inside specialized modality).
Patterns can be preserving (of sort $\PPattern$),
semi-terminating  (of sort $\STPattern$),
 terminating (of sort $\TPattern$),
terminating and preserving (of sort $\TPPattern$), or neither.
The subsort relation is defined by $\TPPattern<\TPattern < \STPattern<\Pattern$ and
$\TPPattern < \PPattern<\Pattern$.
Patterns are constructed with modalities
\begin{equation*}
\Keep{\_}:\NeFactSet\rightarrow\PPattern,\quad\!\!
\Ret{\_}:\NeFactSet\rightarrow\TPPattern,\quad\!\!
\Del{\_}:\NeFactSet\rightarrow\STPattern,\quad\!\!
\Fresh{\_}:\NeFactSet^\FrMark\rightarrow\Pattern.
\end{equation*}
and associative and commutative, subsort overloaded operator:
\begin{gather*}
\_\circ\_ : \Pattern\;\Pattern \rightarrow\Pattern.\quad
\_\circ\_:\Pattern\;\TPattern\rightarrow\TPattern,\quad
\_\circ\_:\PPattern\;\PPattern\rightarrow\PPattern,\\
\_\circ\_:\PPattern\;\TPPattern\rightarrow\TPPattern,\quad
\_\circ\_:\STPattern\;\Pattern\rightarrow\STPattern.
\end{gather*}
Thus, a terminating (resp. a semi-terminating) pattern has to contain at least one
fact wrapped with $\Ret{\_}$ (resp. with either $\Ret{\_}$ or $\Del{\_}$).
A terminating {\em and} preserving pattern
consists of  facts marked only with $\Keep{\_}$ or $\Ret{\_}$, and it contains at least one fact wrapped with $\Ret{\_}$.
Directed  equalities $[F_1]_m\circ [F_2]_m=[F_1\circ F_2]_m$,
where  $m\in\{!,?,0,\FrMark\}$, and $F_1$, $F_2$
are non-empty multisets of (fresh) facts
guarantee that fully reduced patterns have facts gathered in groups of the same modality.

\begin{example}
Let $\sort{Id}$ be a sort, let $f, g:\sort{Id}\;\Nat\rightarrow\Fact$ and $h:\sort{Id}\rightarrow\Fact$
be fact constructors and let $x$, $y$, $z$ and $t$ be variables. Then
\begin{equation*}
[f(x,y)\circ h(x)]_?\circ[g(x, y)]_!\circ[C_{\sort{Id}}(t)]_\FrMark\circ[g(x, 1)]_0
\end{equation*}
is a pattern. It is terminating because of a presence of $[f(x,y)\circ h(x)]_?$ subpattern. It is not, however, preserving since it contains facts wrapped in  $[\_]_0$ and $[\_]_\FrMark$.
\end{example}

The informal meaning of modalities is as follows:
Facts matched by those marked by $\Ret{\_}$  can be considered at most once during quantifier evaluation, but they are not removed from the database.
Facts marked by $\Del{\_}$ are removed from the database when matched, but they are returned
if the computation branch this matching leads to is unsuccessful. Thus, the presence of facts marked by $\Del{\_}$ in the pattern may not guarantee termination, unless one can prove that the formula under quantifier is always successful. Facts marked by $\Keep{\_}$
are always retained in the database, and $\Fresh{\_}$ wraps fresh facts.

\begin{remark}
\label{PatternNotationRemark}
In what follows we use the following notation.
Let $P$ be a pattern. We denote by $P_0$, $P_!$, $P_?$, $P_\FrMark$ the multisets of facts
consisting of those facts in $P$ which are wrapped by modalities $\Del{\_}$,
$\Keep{\_}$, $\Ret{\_}$, and $\Fresh{\_}$, respectively. Thus, e.g.,
$(\Keep{F_1}\circ\Ret{F_2})_?:=F_2$ and  $(\Keep{F_1}\circ\Ret{F_2})_{\FrMark}:=\Empty$.
\end{remark}

\section{Query and condition languages}
This section introduces the three sublanguages of $\QLang$, their syntax and informal semantics.
Formal semantics based on conditional term rewriting is provided in the subsequent sections.

\subsection{Conditions}

The language $\QLang^\cond$ of conditions
on finite multisets of facts is analogous
to first-order logic with quantification restricted to the active domain.

\begin{definition}
Let $\Sigma$ be a signature and let $\cD$ be
a $\Sigma$-algebra of facts.
Formulas of $\QLang^\cond$ are (generally non-ground) terms of sort $\Cond$
constructed with
\begin{gather*}
\False:\rightarrow\Cond,\ \{\_\}:\Bool\rightarrow\Cond,\
\neg\_:\Cond\rightarrow\Cond,\
\_\vee\_:\Cond\;\Cond\rightarrow\Cond,\ \exists\_.\_:\TPPattern\;\Cond\rightarrow\Cond.
\end{gather*}
\end{definition}
Thus, $\False$, $\{B\}$, $\neg\psi$, $\psi\vee\psi'$ and $\exists P \mathbin{.} \psi$
are conditions if $B$ is a term of sort $\Bool$, $P$ is a terminating and preserving pattern,
and $\psi$ and $\psi'$ are conditions.
Consider condition $T:=\exists P\mathbin{.}\psi$. Existential quantifier
$\exists$ binds in $\psi$ all the  variables
 appearing in $P$ which were not bound by the term surrounding $T$. Thus,
the meaning of the formula may change when it is placed in a different context.

\begin{example}
\label{firstexex}
Let $R:\Nat\ \Nat\rightarrow\Fact$, and suppose that terms $R(t_1, t_2)$ represent rows of a relation
$\mathbf{R}\subseteq\mathbf{N}\times\mathbf{N}$.
  Let  $x$, $y$, and $z$ be distinct variables. Then condition
  $\neg\exists\Rett{R(x,y)}.\exists\Rett{R(x,z)}.\neg\{y=z\}$
   expresses functional dependency from the first to the second column of $\mathbf{R}$.
   The first quantifier binds $x$ and $y$, the second one binds $z$.
  The condition is closed. The
  subcondition  $\exists\Rett{R(x,z)}.\neg\{y=z\}$ taken on its own is open,
  but only $y$ is free and the quantifier now binds {\em both} $x$ and $z$.
\end{example}

 Closed formulas in $\QLang^\cond$ are called {\em sentences} in $\QLang^\cond$.
Let $\Var(t)$ be the set of  variables of  $t$, and
let $\closed(\phi)$ iff  $\phi$ in $\QLang^\cond$ is closed.
To define $\closed(\phi)$ by structural recursion we need to keep track
of variables bound by the context of $\phi$. Thus,
$\closed(\phi):=\closed(\phi,\emptyset)$, where,
for any set of variables $V$,
\begin{gather}
\closed(\False, V)=\ttrue,\quad \closed(\neg\phi, V)=\closed(\phi, V),\quad
\closed(\{t\}, V)=\Var(t)\subseteq V\nonumber\\
\closed(\phi_1\vee\phi_2, V)
= \closed(\phi_1, V)\wedge\closed(\phi_2, V),
\quad
\closed(\exists P\mathbin{.}\phi, V)
=\closed\bigl(\phi, V\cup\Var(P)\bigr).
\end{gather}

As the syntactic sugar we define  operators
$\True:\rightarrow\Cond$,
$\_\wedge\_:\Cond\;\Cond\rightarrow\Cond$,
$\forall\_.\_:\TPPattern\;\Cond\rightarrow\Cond$.
with equalities
$\True=\neg\False$,
$\phi_1\wedge\phi_2=\neg(\neg\phi_1\vee\neg\phi_2)$,
$\forall P\mathbin{.}\phi=\neg\exists P\mathbin{.}\neg\phi$.
Later we prove that, for any condition $\phi$,
$\neg\neg\phi$ is logically equivalent to $\phi$. Then the
functional dependency in Example~\ref{firstexex} can be equivalently written as
$
\forall\Rett{R(x,y)}.\forall\Rett{R(x,z)}.\{y=z\}.
$
\begin{definition}
\label{SubCondDef}
A subcondition $\psi$ of $\phi$ in $\QLang^\cond$ is a subterm of $\phi$
of sort $\Cond$.
\end{definition}

\subsection{Syntax of queries and DML queries}

Syntactically $\QLang^\query$ is a restriction of $\QLang^\dml$.
Therefore, we define their syntax jointly  as follows:
\begin{definition}
Queries in
$\QLang^\query$ are terms of sort $\Query$.
DML queries in $\QLang^\dml$ are terms of sort $\DQuery$.
Success assured (DML) queries are terms of sorts $\DQuery^s$.
The sorts are ordered with
$\Fact<\Query<\DQuery$ and $\Fact<\DQuery^s<\DQuery$.
Thus, every query in $\QLang^\query$ can be also interpreted as DML query (which
inserts  but doesn't delete). Every
fact is also a query. Success assured DML queries are DML queries guaranteed to return some facts (or at least $\dummy$).
Terms of sort $\DQuery$ are constructed with
\begin{gather*}
\dummy:\rightarrow\DQuery^s, \quad
 \_\Con\_:\DQuery\;\DQuery\rightarrow\DQuery,\quad
\_\Con\_:\Query\;\Query\rightarrow\Query,\quad
 \_\Con\_:\DQuery^s\;\DQuery\rightarrow\DQuery^s,\nonumber\\
 \_\Con\_:\DQuery\;\DQuery^s\rightarrow\DQuery^s,\quad \Empty:\rightarrow\Query,\quad
\_\Rightarrow\_:\Cond\;\DQuery\rightarrow\DQuery,\quad
  \_\Rightarrow\_:\Cond\;\Query\rightarrow\Query,\nonumber\\
  \From\_.\_:\TPattern\;\DQuery \rightarrow\DQuery,\quad
\From\_.\_:\TPPattern\;\Query \rightarrow\Query,\quad
\From\_.\_:\STPattern\;\DQuery^s \rightarrow\DQuery.
\end{gather*}
\end{definition}
Thus, $\Empty$, $f$, $Q\Con Q'$, $\phi\Rightarrow Q$, and $\From P\mathbin{.}Q$ are queries
in $\QLang^\query$ if $f$ is a fact, $P$ is a terminating and preserving pattern,  $Q$ and $Q'$ are
queries in $\QLang^\query$, and $\phi$ is a condition. Similarly, $\Empty$, $\dummy$,
$f$, $D\Con D'$, $\phi\Rightarrow D$, and $\From P\mathbin{.}D$ are DML queries
in $\QLang^\dml$ if $f$ is a fact, $P$ is a terminating  pattern,  $D$ and $D'$ are
DML queries in $\QLang^\query$, and $\phi$ is a condition.
In addition, $\From  P\mathbin{.}D$ is a DML query if $P$ is a more general semi-terminating
pattern, but $D$ is success-assured, i.e., either $\dummy$, a fact, or
 of the form $D_1\Con D_2$ where at least one of $D_1$, $D_2$ is success assured.
We force $D$ to be success assured when $P$ is semi-terminating, but not terminating
because quantification over semi-terminating pattern may not terminate
 if the quantified expression can fail.
Informal semantics of (DML) queries is given by:
\begin{enumerate}
\item Let $f:\Fact$. A query $f$ returns $f$. A DML query $f$ adds $f$ to the current database.
\item $\dummy$ is used to mark a branch of a DML query as successful even if it does not
add any facts.
\item $\Empty$ is a query returning the empty multiset of facts or a DML query which does nothing.
\item A query $Q\Con Q'$ returns the multiset union of results of $Q$ and $Q'$.
A DML query $D\Con D'$ adds to and removes from the database a multiset union of what $D$ and $D'$ add
and remove, respectively. Since facts are removed immediately,
$\_\Con\_$ is not commutative for DML queries.
\item A query $\phi\Rightarrow Q$ returns what $Q$ returns if $\phi$ is satisfied. It returns $\Empty$
otherwise. A DML query $\phi\Rightarrow D$ does what $D$ does if $\phi$ is satisfied. It
does nothing otherwise.
\item Quantifier $\From P.Q$
 denotes iteration over facts in the database matching $P$. At each iteration step $\sigma(Q)$
is executed, where $\sigma$ is the matching substitution.
If $P$ contains fresh facts,
execution of a DML query $\From  P\mathbin{.}Q$ may introduce fresh values.
\end{enumerate}
Let
$\closed(Q)$ if and only if the query $\phi$ in $\QLang^\dml$ (or $\QLang^\query$) is closed.
Similarly as in the case of conditions,
$\closed(Q):=\closed(Q,\emptyset)$, where,
for any set of variables $V$,
\begin{gather*}
\closed(\Empty, V)=\closed(\dummy, V)=\ttrue,\quad
\closed(\phi\Rightarrow Q, V)=\closed(\phi, V)\wedge\closed(Q, V),\nonumber\\
\closed(f, V)=\Var(f)\subseteq V\ \text{if}\ f:\Fact\nonumber\\
\closed(Q_1\Con Q_2, V)
= \closed(Q_1, V)\wedge\closed(Q_2, V),
\quad
\closed(\From P\mathbin{.}Q, V)
=\closed\bigl(Q, V\cup\Var(P)\bigr).
\end{gather*}

\begin{definition}
A (DML) subquery of a (DML) query $Q$ is a subterm of
$Q$ of sort ($\sort{D}$)$\Query$.
\end{definition}

\begin{example}
The following query in $\QLang^\query$ is closed:
\begin{equation*}
\From [f(x)\circ g]_?\circ[h(x)]_!\mathbin{.}\bigl(\{x > 5\}\Rightarrow f(x + 1)\bigr).
\end{equation*}
It returns a multiset of facts consisting of facts of the form $f(x+1)$ where $f(x)$ and $h(x)$ belong to the multiset we query and $x>5$. When evaluating the query each source fact
$f(x)$ (such that $h(x)$ is in the multiset and $x>5$) and some fact $g$ are matched only once during the evaluation of this query (on the other hand, facts of the form $h(x)$ can be matched multiple times).
Thus, we return at most $n$ facts $f(x+1)$, where $n$ is the number of $g$ facts in the database. For example, for multiset $f(7)\circ f(8)\circ f(7)\circ f(4)\circ h(8)\circ h(7)\circ g\circ g$ the query returns either $f(7)\circ f(8)$ or $f(7)\circ f(7)$ depending on whether we match $f(7)\circ g$ and  $f(8)\circ g$ or we match $f(7)\circ g$ twice (and then, in both cases, we run out of $g$-facts).
\end{example}

\begin{example}
The following (closed) ``DML query'' is not syntactically correct (it cannot be assigned  sort
$\DQuery$):
\begin{equation*}
\From [f(x)]_0\mathbin{.}\bigl(\{x>5\}\Rightarrow g(x)\bigr)
\end{equation*}
This is because the pattern $[f(x)]_0$ is only semiterminating, but not terminating. In this case the subquery in the scope of $\From [f(x)]_!\mathbin{.}\_$ should be success assured, but
$\{x>5\}\Rightarrow g(x)$ clearly isn't (if $x\leq 5$ then it fails).
If we would execute this query against a multiset $f(1)\circ f(6)$ it would loop forever
matching $f(1)$, removing it, failing when executing $\{1>5\}\Rightarrow g(1)$, returning $f(1)$ to the multiset, matching it again, and so on.
\end{example}

As the following example demonstrates, some incorrect DML queries of the form
$\From P\mathbin{.}\psi$, where $P$ is semiterminating but not terminating and $\phi$ is not success assured, provably always terminate. Nevertheless, we prefer to reject them regardless, since usually making them syntactically correct is not difficult:
\begin{example}
Assume that $x$ is a variable of sort $\Nat$ (natural numbers).
Execution of the following syntactically incorrect DML query always terminates
(it replaces each fact of the form $f(x)$ where $x>5$ with fact $g(x)$, and returns to the multiset all other $f(x)$'s unchanged):
\begin{equation*}
\From [f(x)]_0\mathbin{.}\bigl((\{x>5\}\Rightarrow g(x))\Con(\{x\leq 5\}\Rightarrow f(x))\bigr).
\end{equation*}
Termination follows from the fact that
 $(\{x>5\}\Rightarrow g(x))\Con(\{x\leq 5\}\Rightarrow f(x)\bigr)$ always succeeds regardless of which natural number is bound to $x$: depending on whether  $x>5$ or $x\leq 5$ (and one of these must be true) either left or right argument of $\Con$ succeeds and hence the whole subquery succeeds.

\medskip
 Unfortunately, neither $\{x>5\}\Rightarrow g(x)$ nor $\{x\leq 5\}\Rightarrow f(x)$ is syntactically success assured and hence  $(\{x>5\}\Rightarrow g(x))\Con(\{x\leq 5\}\Rightarrow f(x)\bigr)$ is also not success assured. However, it is very easy to modify the above query so that it describes the same modification to the database, but is now syntactically correct and can be assigned sort $\DQuery$:
\begin{equation*}
\From [f(x)]_0\mathbin{.}\bigl((\{x>5\}\Rightarrow g(x))\Con(\{x\leq 5\}\Rightarrow f(x))\Con \Ok\bigr).
\end{equation*}
\end{example}

\section{Rewriting semantics of $\QLang^\cond$}

Semantics of a sentence $\phi$ in $\QLang^\cond$ is given by
the rewriting system $\cR_{\Sigma, \cD}^\cond(\phi)$.
Terms rewritten by $\cR_{\Sigma, \cD}^\cond(\phi)$ are of the form $\{F,S\}^c$, where
 $\{\_,\_\}^c:\FactSet\;\StackC\rightarrow\StateC$, $F$ is the  database of facts
 on which $\phi$ is checked,
and  $S$ is a stack of sort $\StackC$ implementing structural recursion.
  Normal forms, constructed with
$\Sat:\Bool\rightarrow\StateC$, encapsulate the result of $\phi$'s evaluation.
Sort $\StateC$ of terms holding the full state of evaluation must be distinct from all the other sorts and it must not have any super- or sub-sort relation to other sorts. This guarantees that there are no constructors accepting terms of sort $\StateC$ as arguments. Consequently, if all rewrite rules have terms of this sort on the left-hand side, then no subterms can be rewritten, i.e., the defined rewriting system is top-level.
 Terms of sort $\StackC$ are built from frames of sort
$\NodeC$, where $\NodeC<\StackC$, using an associative
binary operator  $\_\_:\StackC\;\StackC\rightarrow\StackC$ with identity $\Empty$.
Most constructors
of frames are indexed by subconditions $\psi$ of $\phi$,
 and lists of distinct variable names
$\vec{v}:=v_1,\ldots, v_n$  of respective sorts $s_1,\ldots, s_n$
($\vec{v}$ can be of any length, even be empty,
as long as $\{\vec{v}\}\subseteq\Var(\phi)$ and
it contains free variables of $\psi$):
\begin{gather*}
\Res:\Bool\rightarrow\NodeC,\quad \Not:\rightarrow\NodeC,\quad
\quad [\_,\ldots,\_]^{\vec{v}| c}_{\psi},[\_,\ldots,\_]^{\vec{v}|c,\downarrow}_{\psi}:
s_1\ldots s_n\rightarrow\NodeC,
\\
[\_|\_,\ldots,\_]^{\vec{v}|c}_{\exists P\mathbin{.}\psi}:\FactSet\;s_1\ldots s_n\rightarrow\NodeC%\
%\text{where}\ \psi=\exists P\mathbin{.}\psi'.
\end{gather*}
As $\vec{v}$ can be empty, the above signature templates include
$[]^{|c}_\psi, []^{|c,\downarrow}_\psi:\rightarrow\NodeC$ and
 $[\_|]^{|c}_\psi:\FactSet\rightarrow\NodeC$.
$\Res(B)$ encapsulates the result of evaluation of a subcondition. Frame
 $\Not$ negates the result of the next frame on the stack.
Frames of the form $[\vec{a}]^{\vec{v}|c}_\psi$, $[\vec{a}]^{\vec{v}|c,\downarrow}_\psi$,
or $[F'|\vec{a}]^{\vec{v}|c}_\psi$
are called $(\psi, \sigma)$-frames, where $\sigma:=\{\vec{a}/\vec{v}\}$
 is the current substitution.
They are related to evaluation of $\sigma(\psi)$.
Marked frames $[\vec{a}]^{\vec{v}|c,\downarrow}_{\psi}$
occur in evaluation of disjunctions $\_\vee\_$.
Iterator frames $[F'|\vec{a}]^{\vec{v}|c}_{\exists P\mathbin{.}\psi}$
represent iterative evaluation of quantifiers.
Multiset $F'$, called iterator state, contains facts available for matching with
$P_!\circ P_?$. Given a database $F$,
to evaluate sentence $\phi$ we rewrite a state term
$\Init_\phi(F):=\{F,[]^{|c}_\phi\}$
until a normal form $\Sat(B)$ is reached. If $B$ then $\phi$ is satisfied in $F$.

\medskip
 Now we present rule schemata for $\cR_{\Sigma, \cD}^\cond(\phi)$ instantiated for a given formula
$\phi$ in $\QLang$. It is important to distinguish between object variables
substituted when applying actual rules, and
metavariables used to define rule templates. Below,
$F, F':\FactSet$, $S:\StackC$,  and $B:\Bool$ are object variables. We also denote
by $\vec{v}:=v_1, \ldots, v_n$ and $\vec{w}=w_1,\ldots,w_m$
 sequences of object variables in $\Var(\phi)$. Metavariables,
$\psi$, $\psi_1$, and $\psi_2$ stand for arbitrary subconditions of $\phi$, metavariable $\cB$
stands for Boolean subterms of $\phi$.  Metavariable $P$
stands for arbitrary patterns in $\phi$, and
$P_!$ and $P_{?}$ denote multisets of facts in the instance of pattern
$P$ marked by respective modalities
(see Remark~\ref{PatternNotationRemark}).

\medskip
Constant $\False$ evaluates to $\ffalse$, and
$\{\cB\}$ evaluates to $\sigma(\cB)$, where $\sigma$ is the current substitution of variables $\vec{v}$
 ($\sigma(\cB)$ is ground since $\{\vec{v}\}$ contains all free variables (i.e., all variables) of $\cB$, hence, by assumption, it
 simplifies to either $\ffalse$ or $\ttrue$):
 \begin{equation}
\LFalse:\{F,S [\vec{v}]_\False^{\vec{v}|c}\}^c\Rightarrow\{F, S\Res(\ffalse)\}^c,
 \quad\quad
 \LBool:\{F, S[\vec{v}]_{\{\cB\}}^{\vec{v}|c}\}^c\Rightarrow\{F, S\Res(\cB)\}^c.
 \label{Bool-schema}
\end{equation}
To evaluate $\sigma(\neg\psi)$, we ``unfold'' by replacing the $(\neg\psi, \sigma)$-frame with
$\neg$ and $(\psi,\sigma)$-frame. When $\sigma(\psi)$ is evaluated we ``fold'' by negating the result:
\begin{equation}
\LNotUnfl :\{F,S [\vec{v}]_{\neg\psi}^{\vec{v}|c}\}\Rightarrow\{F,S\Not[\vec{v}]_{\psi}^{\vec{v}|c}\}^c,\quad\quad
\LNotFld:\{F,S\Not\Res(B)\}\Rightarrow\{F,S\Res(\neg B)\}^c.
\label{NotEq}
\end{equation}

\begin{remark}
\label{MetaRem} In our notation the same symbols often play a dual role --- as subterms, and as part of
function symbols (in sub- and super-scripts).
Consider the following instantiation of schema $\LNotUnfl$:
\begin{equation*}
\LNotUnfl :\{F, S[x,y]^{x,y|c}_{\neg\{x=y\}}\}^c\Rightarrow
\{F, S\Not[x,y]^{x,y|c}_{\{x=y\}}\}^c
\end{equation*}
Variables in sub- and super-scripts are never substituted:
with the above rule we have a one step rewrite of ground terms
$
\{\Empty,[0,1]^{x,y|c}_{\neg\{x=y\}}\}^c\xrightarrow{\LNotUnfl }
\{\Empty,\Not[0,1]^{x,y|c}_{\{x=y\}}\}^c.
$
In schema $\LBool$ metavariable $\cB$ occurs both as  part of a name
and as a subterm. An instantiation
$
\LBool:\{F,S[x,y]^{x,y|c}_{\{x=y\}}\}^c\Rightarrow\{F, S\Res(x=y)\}^c
$  yields (with $S=\Not$, $x=0$, and $y=1$) a one step rewrite
$
\{\Empty,\Not[0,1]^{x,y|c}_{\{x=y\}}\}^c\xrightarrow{\LBool}
\{\Empty, \Not\Res(0=1)\}^c=\{\Empty, \Not\Res(\ffalse)\}^c.
$
\end{remark}

 To evaluate
 disjunction $\psi_1\vee\psi_2$ we create two frames corresponding to the disjuncts.
  If $\psi_2$ evaluates to
 $\ttrue$, the frame marked by $\downarrow$ is dropped (disjunctions are short circuited). If $\psi_2$
 evaluates to $\ffalse$, the frame corresponding to $\psi_1$ drops $\downarrow$ and is evaluated normally.
 \begin{gather}
 \DisjUnfl:\{F, S[\vec{v}]_{\psi_1\vee\psi_2}^{\vec{v}|c}\}^c\Rightarrow
 \{F,S[\vec{v}]_{\psi_1}^{\vec{v}|c,\downarrow}[\vec{v}]_{\psi_2}^{\vec{v}|c}\}^c,\nonumber\\
 \DisjFldt:\{F,S[\vec{v}]_{\psi}^{\vec{v}|c,\downarrow}\Res(\ttrue)\}^c\Rightarrow
 \{F,S\Res(\ttrue)\}^c,\quad
  \DisjFldf:\{F,S[\vec{v}]_{\psi}^{\vec{v}|c,\downarrow}\Res(\ffalse)\}^c\Rightarrow
 \{F,S[\vec{v}]_{\psi}^{\vec{v}|c}\}^c,\quad
 \label{DisjEq}
 \end{gather}

Quantifier evaluation is initialized with the whole database available for matching:
\begin{equation}
\QuantInit:\{F, S[\vec{v}]^{\vec{v}|c}_{\exists P\mathbin{.}\psi}\}^c\Rightarrow
\{F, S[F\;|\;\vec{v}]^{\vec{v}|c}_{\exists P\mathbin{.}\psi}\}^c.
\end{equation}
Let  $\vec{w}$ be a sequence of
 all the distinct variables in $\Var(P)\setminus\{\vec{v}\}$. Let $\sigma$ be the current substitution.
Rule $\QuantUnfl$ pushes onto
the stack a $(\sigma',\psi)$-frame, where $\sigma'=\sigma\cup\{\vec{b}/\vec{w}\}$
is defined by matching $F'\circ\sigma(P_!\circ P_?)$ with iterator state, and it removes $\sigma'(P_?)$
from the iterator state.
We keep applying $\QuantUnfl$ until $\sigma'(\psi)$ evaluates to $\ttrue$ or we cannot match
$\sigma(P_!\circ P_?)$ with iterator state:
 \begin{gather}
 \QuantUnfl:
 \bigl\{F,S[F'\circ P_!\circ P_{?}\;|\;\vec{v}]^{\vec{v}|c}_{\exists P\mathbin{.}\psi}\bigr\}^c
 \Rightarrow
  \bigl\{F,S[F'\circ P_!\;|\;\vec{v}]^{\vec{v}|c}_{\exists P\mathbin{.}\psi}
  [\vec{v},\vec{w}]^{\vec{v},\vec{w}|c}_{\psi}\bigr\}^c,\\
    \QuantFldf:\bigr\{F,S[F'|\vec{v}\bigr]^{\vec{v}|c}_{\exists P\mathbin{.}\psi}
  \Res(\ffalse)\bigr\}^c\Rightarrow
  \bigl\{F,S[F'|\vec{v}]^{\vec{v}|c}_{\exists P\mathbin{.}\psi}\bigr\}^c,\quad
  \QuantFldt:\bigl\{F,S[F'|\vec{v}\bigr]^{\vec{v}|c}_{\exists P\mathbin{.}\psi}
  \Res(\ttrue)\bigr\}^c\Rightarrow
  \{F,S\Res(\ttrue)\bigr\}^c.\nonumber
 \end{gather}
Let $\Succ$ be a sort and let $\yes:\rightarrow\Succ$. For all $\vec{v}\subseteq\Var(\phi)$ and
patterns $P$ occurring in $\phi$ we define a function
$\mu_{P, \vec{v}}:\TPPattern\;s_1\ldots s_n\rightarrow\Succ$ with the single equation
 \begin{equation}
 \label{EndFuncDefEq}
 \mu_{P,\vec{v}}(F'\circ P_!\circ P_?,\vec{v})=\yes.
\end{equation}
Thus, $\mu_{P,\vec{v}}(F,\vec{a})=\yes$ if and only if  $F$ matches with $F'\circ\{\vec{a}/\vec{v}\}(P_?\circ P_0)$.
Since facts matched by $\sigma(P_?)$ are removed from the iteration state, $\QuantUnfl$ cannot
be applied infinitely many times. The following rule schema makes $\sigma(\exists P\mathbin{.}\psi)$
evaluate to $\ffalse$ when
$\QuantUnfl$ can no longer be applied:
\begin{equation}
\label{SatEq}
\QuantEnd:\bigl\{F,S[F'\;|\;\vec{v}]^{\vec{v}|c}_{\exists P\mathbin{.}\psi}\}^c
\Rightarrow\{F,S\Res(\ffalse)\}^c\quad\text{if}\ (\mu_{P,\vec{v}}(F', \vec{v})=\yes)=\ffalse.
\end{equation}

Finally, the rule $\LSat:\{F,\Res(B)\}^c\Rightarrow\Sat(B)$ finishes evaluation of $\phi$.

\begin{theorem}
\label{CondTerminating}
$\cR_{\Sigma, \cD}^\cond(\phi)$ is a terminating rewriting system.
\end{theorem}
\begin{proof}
It suffices to define a partial well-order $\_<_c\_$ on terms of sort
$\StateC$ which makes rewriting strictly monotonic, i.e., such that
  $t_1\rightarrow t_2$ implies $t_2<_ct_1$
for all  $t_1, t_2:\StateC$.
The  order is defined by
\begin{equation*}
\Sat(B)<_c\{F,S\}^c,\quad \{F,S\}^c<_c\{F,S'\}^c\ \text{iff}\ S<_sS',
\end{equation*}
for all $F$, $S$, $S'$, $B$.
Here $\_<_s\_$ is the lexicographic order on stacks derived from partial order  $\_<_f\_$
on frames, i.e., for all frames $D_1,\ldots, D_n$, $E_1,\ldots,E_m$,
$D_1\ldots D_n<_sE_1\ldots E_m$ iff either (1)
for some $k$, $D_k<_f E_k$ and $D_i=E_i$ for $i\in\{1,\ldots,k-1\}$,
or (2) $n<m$ and $D_i=E_i$ for $i\in\{1,\ldots,n\}$.
Frame ordering is defined by
$
\Res(B)
<_f\Not
<_f[F\;|\;\vec{a}]^{\vec{v}|c}_\psi
<_f[F'\;|\;\vec{a}]^{\vec{v}|c}_\psi
<_f[\vec{b}]_\psi^{\vec{w}|c}
<_f[\vec{b}]_\psi^{\vec{w}|c,\downarrow}
<_f[F''\;|\;\vec{c}]^{\vec{x}|c}_{\psi'}
$,
for all $F$, $F'$, $F''$, $\psi$, $\psi'$, $B$, $\vec{v}$,
$\vec{w}$, $\vec{x}$, $\vec{a}$, $\vec{b}$, $\vec{c}$ such that
 $F\subsetneq F'$ and $\psi$ is a proper subcondition of $\psi'$.
Partial order $\_<_f\_$ is Noetherian because multisets of facts $F$, $F'$ and conditions $\psi$, $\psi'$
are  finite terms.
Hence, if the stacks are of bounded size, also $\_<_c\_$ is Noetherian.
The size of stacks is bounded because each stack size increasing rule is of the form
$\{F, S[\ldots]^{\ldots}_{\psi_1}\}^c\rightarrow\{F, SA[\ldots]^{\ldots}_{\psi_2}\}^c$,
where $A$ is a frame and $\psi_2$ is a proper subterm of $\psi_1$.
Since rules in $\cR_{\Sigma, \cD}^\cond(\phi)$ are topmost, the  rewriting is strictly monotonic
because $t_2<_c t_1$ for each rule schema
$t_1\Rightarrow t_2\;\text{if}\;C$ in $\cR_{\Sigma,\cD}^\cond(\phi)$.
\end{proof}

The following useful observation can be trivially verified by examining the rule schemas:
\begin{lemma}
\label{SubLemma}
Let $\psi$ be a subcondition of $\phi$. For any finite multiset of facts $F$, stack $S$,
Boolean $B$, variables $\vec{v}=v_1,\ldots,v_n$ and values $\vec{a}=a_1,\ldots, a_n$,
 $\{F, S[\vec{a}]^{\vec{v}|c}_\psi\}^c\rightarrow^*\{F, S\Res(B)\}^c$
in $\cR_{\Sigma, \cD}^\cond(\phi)$
iff
$\{F,[]^{|c}_{\sigma(\psi)}\}^c\rightarrow^*\Sat(B)$ in $\cR_{\Sigma, \cD}^\cond(\sigma(\psi))$,
where substitution $\sigma:=\{\vec{a}/\vec{v}\}$.
\end{lemma}

The following example shows verification of a condition in $\QLang^\cond$ for a given multiset of
facts using rewriting semantics. It also shows non-confluence of the resulting rewrite system.
\begin{example}
\label{NonConflEx}
Suppose $p, q:\Nat\rightarrow\Fact$. Let $\psi:=\exists\;\Rett{q(z)}\mathbin{.}\{z=x\}$,
$K:=p(1)\circ q(0)$, $L:=p(0)\circ q(0)$, and $M:=p(0)\circ p(0)\circ p(1)$. To check if condition
\begin{equation*}
\phi:=\exists [p(x)\circ p(y)]_{?}\circ[q(y)]_!\mathbin{.}
\bigl(\neg\psi\vee\{x=y\}\bigr).
\end{equation*}
is satisfied in a multiset
$H:=p(0)\circ p(0)\circ K=L\circ p(0)\circ p(1)=M\circ q(0)$ we normalize
\begin{align*}
\Init_\phi(H)\xrightarrow{\QuantInit}&
\bigl\{H,[H|\;]^{|c}_\phi\bigr\}^c
\xrightarrow{\QuantUnfl}\bigl\{H,
[K|\;]^{|c}_\phi[0,0]^{x,y|c}_{\neg\psi\vee\{x=y\}}\bigr\}^c
\xrightarrow{\DisjUnfl}\bigl\{H,
[K|\;]^{|c}_\phi[0,0]^{x,y|c,\downarrow}_{\neg\psi}
[0,0]^{x,y|c}_{\{x=y\}}\bigr\}^c\\
\xrightarrow{\LBool}&\bigl\{H,
[K|\;]^{|c}_\phi[0,0]^{x,y|c,\downarrow}_{\neg\psi}
\Res(0=0)\bigr\}^c
\xrightarrow{\DisjFldt}\bigl\{H,
[K|\;]^{|c}_\phi\Res(\ttrue)\bigr\}^c
\xrightarrow{\QuantFldt}\bigl\{H,\Res(\ttrue)\bigr\}^c
\xrightarrow{\LSat}\Sat(\ttrue).
\end{align*}
Thus, $\phi$ is satisfied in $H$.
However, we have also a normalizing sequence ending with $\Sat(\ffalse)$:
\begin{align*}
\Init_\phi(H)\xrightarrow{\QuantInit}&
\bigl\{H,[H\;|\;]^{|c}_\phi\bigr\}^c
\xrightarrow{\QuantUnfl}\bigl\{H,
[L\;|\;]^{|c}_\phi[0,1]^{x,y|c}_{\neg\psi\vee\{x=y\}}\bigr\}^c
\rightarrow^*\bigl\{H,
[L\;|\;]^{|c}_\phi[0,1]^{x,y|c}_{\neg\psi}\bigr\}^c\\
\rightarrow^*&\bigl\{H,
[L\;|\;]^{|c}_\phi\Not
[H\;|\;0,1]^{x,y|c}_{\psi}\bigr\}^c
\xrightarrow{\QuantUnfl}\bigl\{H,
[L\;|\;]^{|c}_\phi\Not
[M\;|\;0,1]^{x,y|c}_{\psi}
[0,1,0]^{x,y,z|c}_{\{z=x\}}\bigr\}^c\\
\rightarrow^*&\bigl\{H,
[L\;|\;]^{|c}_\phi\Not\Res(\ttrue)\bigr\}^c
\rightarrow^!\Sat(\ffalse).
\end{align*}
\end{example}

Thus, evaluation of conditions does not necessarily lead to a unique result
 (the rewriting system is not confluent).
This requires making the definition of logical equivalence bisimulation-like:

\begin{definition}
Let $\phi_1$ and $\phi_2$ be conditions in $\QLang^\cond$.
We say that $\phi_1$ is logically equivalent to $\phi_2$, writing $\phi_1\equiv\phi_2$, if and only if
for all ground multisets of facts $F$, ground substitutions $\sigma$ such that $\sigma(\phi_1)$
and $\sigma(\phi_2)$ are closed, and a Boolean $B\in\{\ttrue,\ffalse\}$
we have
\begin{equation*}
\Init_{\sigma(\phi_1)}(F)\rightarrow^!\Sat(B)\quad\text{if and only if}\quad \Init_{\sigma(\phi_2)}(F)\rightarrow^!\Sat(B).
\end{equation*}
\end{definition}

The following result is an immediate consequence of Lemma~\ref{SubLemma}:
\begin{lemma}
Logical equivalence on conditions in $\QLang^\cond$ is an equivalence and a congruence, i.e., if
$\kappa$ is a position in a condition $\phi$ in $\QLang^\cond$ such that $\phi|_\kappa$ is
a condition, and $\psi\equiv \phi|_\kappa$, then $\phi\equiv\phi[\psi]_\kappa$.
\end{lemma}

A renaming is an injective substitution $\sigma$ mapping variables to variables.
\begin{lemma}
For any closed condition $\phi$ in $\QLang^\cond$, and any renaming $\sigma$,
$\phi\equiv\sigma(\phi)$.
\end{lemma}

The following result clarifies elements of rewriting semantics of sentences in $\QLang^\cond$:
\begin{lemma}
\label{LemmaClearSem}
For each ground multiset of facts $F$, and all sentences $\phi$, $\phi_1$,
$\phi_2$ and $\exists P\mathbin{.}\psi$
:
\begin{enumerate}
\item $\Init_\phi(F)\rightarrow^!\Sat(\ttrue)$ or $\Init_\phi(F)\rightarrow^!\Sat(\ffalse)$,
and these are the only possible normal forms of $\Init_\phi(F)$.
\item $\Init_\False(F)\rightarrow^!\Sat(\ffalse)$ and never $\Init_\False(F)\rightarrow^!\Sat(\ttrue)$.
\item $\Init_\phi(F)\rightarrow^!\Sat(B)$ iff $\Init_{\neg\phi}(F)\rightarrow^!\Sat(\neg B)$
for all $B:\Bool$.
\item $\Init_{\phi_1\vee\phi_2}(F)\rightarrow^!\Sat(\ttrue)$ iff
$\Init_{\phi_1}(F)\rightarrow^!\Sat(\ttrue)$ or
$\Init_{\phi_2}(F)\rightarrow^!\Sat(\ttrue)$.
\item $\Init_{\phi_1\vee\phi_2}(F)\rightarrow^!\Sat(\ffalse)$ iff
$\Init_{\phi_1}(F)\rightarrow^!\Sat(\ffalse)$ and
$\Init_{\phi_2}(F)\rightarrow^!\Sat(\ffalse)$.
\item $\Init_{\exists P\mathbin{.}\psi}(F)\rightarrow^!\Sat(\ttrue)$ iff
there exists a substitution $\sigma$, and a multiset  $F'$ such that
$F'\circ\sigma(P_{?}\circ P_!)=F$ and $\Init_{\sigma(\psi)}(F)\rightarrow^!\Sat(\ttrue)$.
\item $\Init_{\exists P\mathbin{.}\psi}(F)\rightarrow^!\Sat(\ffalse)$ iff
there exist two sequences of ground
multisets of facts $F_0, F_1,\ldots,F_n$ and $G_0,G_1,\ldots, G_{n-1}$,
 and a sequence
of substitutions $\sigma_0,\sigma_1,\ldots,\sigma_{n-1}$ such that
\begin{enumerate}
\item  $F_0=F$,
$F_{i+1}=G_i\circ \sigma_i(P_!)$, and $F_i=G_i\circ\sigma_i(P_!\circ P_{?})$,
for all $i\in\{0,\ldots, n-1\}$,
\item
$\Init_{\sigma_i(\psi)}(F)\rightarrow^!\Sat(\ffalse)$, for all $i\in\{0,\ldots, n-1\}$,
\item there exists no substitution $\sigma_n$ and multiset of facts $G_n$ such that
$F_n=G_n\circ\sigma_n(P_!\circ P_{?})$.
\end{enumerate}
\item $\Init_{\phi\vee\neg\phi}(F)\rightarrow^!\Sat(\ttrue)$.
If both  $\Init_\phi(F)\rightarrow^!\Sat(\ttrue)$ and
$\Init_\phi(F)\rightarrow^!\Sat(\ffalse)$  then also
$\Init_{\phi\vee\neg\phi}(F)\rightarrow^!\Sat(\ffalse)$.
\end{enumerate}
\end{lemma}
\begin{proof}
The first point is  verified by structural recursion using Lemma~\ref{SubLemma}
and rules in Equations~\eqref{Bool-schema}--\eqref{SatEq}. Points 2--7 are
verified using rules in Equations~\eqref{Bool-schema}--\eqref{SatEq}. Point 8 is verified using
points 3-5.
\end{proof}

\begin{lemma}
The following logical equivalences hold between conditions in $\QLang^\cond$:
\begin{gather*}
\phi\vee\False\equiv\phi,\quad
\phi_1\vee\phi_2\equiv\phi_2\vee\phi_1,\quad
\phi_1\vee(\phi_2\vee\phi_3)\equiv (\phi_1\vee\phi_2)\vee\phi_3,\\
\neg\neg\phi\equiv\phi,\quad
\exists P\mathbin{.}(\phi_1\vee\phi_2)\equiv
(\exists P\mathbin{.}\phi_1)\vee
(\exists P\mathbin{.}\phi_2).
\end{gather*}
\end{lemma}
\begin{proof}
The above equivalences can be proven
using points~1-7 in Lemma~\ref{LemmaClearSem}. Only the
last equivalence's proof
is non-trivial.
Let $F$ be a ground multiset of facts and let $\sigma$ be a ground substitution such that
$\sigma(\exists P\mathbin{.}(\phi_1\vee\phi_2))$ (or, equivalently,
$\sigma((\exists P\mathbin{.}\phi_1)\vee
(\exists P\mathbin{.}\phi_2))$) is closed.
Denote $Q:=\sigma(P)$, $\psi_i:=\sigma(\phi_i)$, for $i\in\{1,2\}$.
Using Lemma~\ref{LemmaClearSem}, p.~6, we see that
$\Init_{\exists Q\mathbin{.}(\psi_1\vee\psi_2)}(F)\rightarrow^!\Sat(\ttrue)$
iff  there exists a substitution $\sigma'$, and a multiset of facts $F'$ such that
$F'\circ\sigma'(Q_{?}\circ Q_!)=F$ and
$\Init_{\sigma'(\psi_1)\vee\sigma'(\psi_2)}(F)\rightarrow^!\Sat(\ttrue)$.
The latter holds iff there exists $i\in\{1,2\}$ such that
$\Init_{\sigma'(\psi_i)}(F)\rightarrow^!\Sat(\ttrue)$, by Lemma~\ref{LemmaClearSem}, p.~4.
It follows, again using Lemma~\ref{LemmaClearSem}, point~6,  that
$\Init_{\exists Q\mathbin{.}(\psi_1\vee\psi_2)}(F)\rightarrow^!\Sat(\ttrue)$
iff
$\Init_{\exists Q\mathbin{.}\psi_i}(F)\rightarrow^!\Sat(\ttrue)$
for some $i\in\{1,2\}$, i.e., iff (by Lemma~\ref{LemmaClearSem}, p.~4)
$\Init_{(\exists Q\mathbin{.}\psi_1)\vee(\exists Q\mathbin{.}\psi_2)}(F)\rightarrow^!\Sat(\ttrue)$.
The part of the proof with falsity is more complex. Using Lemma~\ref{LemmaClearSem}, p.~7,
we see that $\Init_{\exists Q\mathbin{.}(\psi_1\vee\psi_2)}(F)\rightarrow^!\Sat(\ffalse)$
iff
there exist two sequences of ground
multisets of facts $F_0,\ldots,F_n$ and $G_0,\ldots, G_{n-1}$,
 and a sequence
of substitutions $\sigma_0,\ldots,\sigma_{n-1}$ such that
(a) $F_0=F$, $F_{i+1}=G_i\circ \sigma_i(Q_!)$ and $F_i=G_i\circ\sigma_i(Q_!\circ Q_{?})$,
for all $i\in\{0,\ldots, n-1\}$,
(b) $\Init_{\sigma_i(\psi_1)\vee\sigma_i(\psi_2)}(F)\rightarrow^!\Sat(\ffalse)$, for all $i\in\{0,\ldots, n-1\}$, (c)
there exists no substitution $\sigma_n$ and multiset of facts $G_n$ such that
$F_n=G_n\circ\sigma_n(Q_!\circ Q_{?})$. Then by Lemma~\ref{LemmaClearSem}, p.~5,
(b) iff, for $j\in\{1,2\}$, (b${}_j$) $\Init_{\sigma_i(\psi_j)}(F_i)\rightarrow^!\Sat(\ffalse)$, for all $i\in\{0,\ldots, n-1\}$. Then (by
Lemma~\ref{LemmaClearSem}, p.~7) (a), (b$_1$), (b$_2$) and (c)
iff
$\Init_{\exists Q\mathbin{.}\psi_j}(F)\rightarrow^!\Sat(\ffalse)$
for $j\in\{1,2\}$ iff, by Lemma~\ref{LemmaClearSem}, p.~5,
$\Init_{(\exists Q\mathbin{.}\psi_1)\vee (\exists Q\mathbin{.}\psi_2)}(F)\rightarrow^!\Sat(\ffalse)$.
\end{proof}

Non-confluence of $\cR_{\Sigma, \cD}^\cond(\phi)$ in
Example~\ref{NonConflEx} depended on  patterns in $\phi$ with more than one fact.
It turns out that $\cR_{\Sigma, \cD}^\cond(\phi)$ may be non-confluent
even if $\phi$ contains only  single-fact patterns:

\begin{example}
Let
$r:\Nat\;\Nat\rightarrow\Fact$ be commutative. Consider condition
$\phi:=\exists \Rett{r(x,y)}\mathbin{.}\{x=1\}$ evaluated in a database $F=r(1,2)$. Since  $r$ is commutative, there are two distinct substitutions
$\{1/x,2/y\}$ and $\{2/x,1/y\}$ which match $r(x,y)$ with $r(1,2)$. Consequently, there are
two distinct paths of evaluating $\phi$:
$
\Init_\phi(F)\rightarrow^*
\bigl\{F,[F|]_{\phi}^{|c}\bigr\}^c
\xrightarrow{\QuantUnfl}
\left\{
\begin{array}{l}
\bigl\{F,[\Empty|]_{\phi}^{|c}[1,2]^{x,y|c}_{\{x=1\}}\bigr\}^c
%\xrightarrow{Bool}\bigl\{F,[\Empty|]_{\phi}^{|c}\Res(\ttrue)\bigr\}^c
\rightarrow^!\Sat(\ttrue)\\
\bigl\{F,[\Empty|]_{\phi}^{|c}[2,1]^{x,y|c}_{\{x=1\}}\bigr\}^c
%\xrightarrow{Bool}\bigl\{F,[\Empty|]_{\phi}^{|c}\Res(\ffalse)\bigr\}^c
\rightarrow^!\Sat(\ffalse)
\end{array}
\right..
$
\end{example}

\begin{definition}
\label{UniqMatchPropDef}
A fully reduced term $t:\Fact$
is said to have a {\em unique matching} property iff  for any ground, fully reduced
 term $t':\Fact$ there exists at most one
substitution $\sigma$ such that $\sigma(t)=_At'$.
\end{definition}

\begin{definition}
\label{NonDeterministicCondDef}
A condition $\phi$ in $\QLang^\cond$ is called {\em deterministic} if and only if
 all quantification patterns in $\phi$ contain
only single facts with unique matching property.
\end{definition}

The following theorem states that while evaluation of a deterministic condition is not itself deterministic, but its results are.

\begin{theorem}
\label{ConfluentPartCond}
Let $\phi$ be a deterministic condition in $\QLang^\cond$. Then
$\cR_{\Sigma,\cD}^\cond(\phi)$ is confluent. In particular, given a ground multiset of facts $F$,
there is a unique $B\in\{\ttrue,\ffalse\}$ such that
$\Init_\phi(F)\rightarrow^!\Sat(B)$.
\end{theorem}
\begin{proof}
We argue by induction on the complexity of  formulas indexing frames on the top of  a stack.
First observe that  only terms of the form $t:=\bigl\{F, S[F'\;|\;\vec{a}]^{\vec{v}|c}_{\exists P\mathbin{.}\psi}\bigr\}^c$ can be rewritten in a single step into
two distinct terms. As semiconfluence implies confluence it suffices to prove that if
$t'\leftarrow t\rightarrow^{+} t''$
then there exists $s$ such that
$t'\rightarrow^* s\;{}^*\!\!\leftarrow t''$.
By Lemmas~\ref{SubLemma} and \ref{LemmaClearSem}, p.~1, a rewrite sequence
$t\rightarrow^{+}t''$ must (1) contain $\{F,S\Res(B)\}^c$, or (2) can be extended to the sequence
ending in $\{F,S\Res(B)\}^c$. If we prove that $t'\rightarrow^*\{F,S\Res(B)\}^c$ then
we can set $s=t''$ (if (1)) or $s=\{F,S\Res(B)\}^c$ (if (2)).
It remains to prove that $t'\rightarrow^*\{F,S\Res(B)\}^c$.
Under the theorem's
assumption, $P=\Rett{f}$, where $f$ is a fact with
a unique matching property (Definition~\ref{UniqMatchPropDef}).
Thus, if for some fact $g$ in $F'$ there exists  a substitution
 $\sigma_g$ extending
$\{\vec{a}/\vec{v}\}$ such that $g=\sigma_g(f)$ and
$\Init_{\sigma_g(\psi)}(F)\rightarrow\Res(\ttrue)$
(which implies, by inductive assumption, that $\Init_{\sigma_g(\psi)}(F)\not\rightarrow^*\Res(\ffalse)$),
 it is unique, and cannot be missed
during evaluation of the iterator frame.
Either
such $g$ exists, and then $t\rightarrow^*\{F,S\Res(\ttrue)\}^c$
but $t\not\rightarrow^*\{F,S\Res(\ffalse)\}^c$ (hence necessarily
both $B=\ttrue$ and $t'\rightarrow^*\{F,S\Res(\ttrue)\}^c$),
or it doesn't, and hence,
both $B=\ffalse$ and $t'\rightarrow^*\{F,S\Res(\ffalse)\}^c$
\end{proof}

\begin{example}
\label{DeterministicNonconfluentEx}
Suppose $p:\Nat\rightarrow\Fact$, $q:\rightarrow\Fact$, and let $r:\Nat\;\Nat\rightarrow\Fact$ be a commutative operator (i.e., $r(x, y)=_A r(y, x)$). Clearly, $p(x)$ and $q$ have unique matching property
(Definition~\ref{UniqMatchPropDef}), while $r(x, y)$ does not, since, e.g., if
$\sigma_1:=\{0/x, 1/y\}$ and $\sigma_2:=\{1/x, 0/y\}$ then $\sigma_1(r(x, y))=_A\sigma_2(r(x,y))=_Ar(0,1)$. Let
\begin{equation*}
\phi_1:=\exists [p(x)]_{?}\mathbin{.}\{x=0\},\quad
\phi_2:=\exists [p(x)\circ q]_{?}\mathbin{.}\{x=0\},\quad
\phi_3:=\exists [r(x, y)]_{?}\mathbin{.}\{x=0\}.
\end{equation*}
Then $\phi_1$ is deterministic (Definition~\ref{NonDeterministicCondDef}), while $\phi_2$ and
$\phi_3$ are not deterministic.
 Let $F:=p(0)\circ p(1)\circ q\circ r(0,1)$. We now consider evaluation of all the $\phi_i$'s on $F$. First, the reader will easily verify that while evaluating the existential quantifier in $\phi_1$ we can either first match $p(x)$ with $p(1)$ and then, upon failure, with $p(0)$, or first match with $p(0)$. Eventually, both paths yield satisfaction of $\phi_1$ on $F$ (although the first path is is longer). On the other hand, evaluation of $\phi_2$ and $\phi_3$ demonstrate two ways in which non-determinism occurs in evaluation of conditions in $\QLang^\cond$.
First, consider two non-convergent paths of rewriting $I_{\phi_2}(F)$:
\begin{equation*}
\xymatrix{
\bigl\{F,[p(1)\circ r(0,1)|\;]^{|c}_{\phi_2}[0]^{x|c}_{\{x=0\}}\bigr\}^c\ar[r]^-{\LBool}&
\bigl\{F,[p(1)\circ r(0,1)|\;]^{|c}_{\phi_2}\Res(\ttrue)\bigr\}^c\ar[r]^-{*}&\Sat(\ttrue)\\
\bigl\{F,[p(0)\circ p(1)\circ q\circ r(0,1)|\;]^{|c}_{\phi_2}\bigr\}^c\ar[u]^{\QuantUnfl}\ar[d]_{\QuantUnfl}&I_{\phi_2}(F) \ar[l]_-\QuantInit&\Sat(\ffalse)\\
\bigl\{F,[p(0)\circ r(0,1)|\;]^{|c}_{\phi_2}[1]^{x|c}_{\{x=0\}}\bigr\}^c\ar[r]^-{\LBool}&
 \bigl\{F,[p(0)\circ r(0,1)|\;]^{|c}_{\phi_2}\Res(\ffalse)\bigr\}^c\ar[r]^-{\QuantEnd}&
\bigl\{F,\Res(\ffalse)\bigr\}^c\ar[u]^-{*}
}
\end{equation*}
Here the reason for non-determinism which ultimately leads to non-convergent paths of execution
 is that when evaluating the existential quantifier at each attempt we have to consume both $p(x)$-fact and $q$-fact: since there is only one $q$ fact, if we start from wrong $p(x)$-fact (i.e., $p(1)$), we do not get the second chance.

 Denote for brevity $K:=p(0)\circ p(1)\circ q$. Recall that $r(x,y)=_Ar(y,x)$. In particular,
 $r(0,1)=_Ar(1,0)$.
 Consider now two non-convergent paths of rewriting $I_{\phi_3}(F)$:
 \begin{equation*}
\xymatrix{
\bigl\{F,[K|\;]^{|c}_{\phi_3}[0,1]^{x,y|c}_{\{x=0\}}\bigr\}^c\ar[r]^-{\LBool}&
\bigl\{F,[K|\;]^{|c}_{\phi_3}\Res(\ttrue)\bigr\}^c\ar[r]^-{*}&\Sat(\ttrue)\\
\bigl\{F,[K\circ r(0,1)|\;]^{|c}_{\phi_3}\bigr\}^c\ar[u]^{\QuantUnfl}\ar[d]_{\QuantUnfl}&I_{\phi_3}(F) \ar[l]_-\QuantInit&\Sat(\ffalse)\\
\bigl\{F,[K|\;]^{|c}_{\phi_3}[1,0]^{x,y|c}_{\{x=0\}}\bigr\}^c\ar[r]^-{\LBool}&
 \bigl\{F,[K|\;]^{|c}_{\phi_3}\Res(\ffalse)\bigr\}^c\ar[r]^-{\QuantEnd}&
\bigl\{F,\Res(\ffalse)\bigr\}^c\ar[u]^-{*}
}
\end{equation*}
Thus, in this case the reason of non-determinism leading to non-convergent paths was the possibility of
two distinct matchings of $r(0,1)$ with $r(x, y)$ given by $\{0/x, 1/y\}$ and $\{1/x, 0/y\}$.
\end{example}

\section{Rewriting semantics of $\QLang^\query$}

Let $Q$ be a query in $\QLang^\query$.
We associate with $Q$ the rewriting system $\cR_{\Sigma, \cD}^\query(Q)$.
Terms rewritten with the rules of the rewriting system $\cR_{\Sigma, \cD}^\query(Q)$
 are of the form $\{F,F',S\}^q$, where
 $\{\_,\_,\_\}^q:\FactSet\;\FactSet\;\StackQ\rightarrow\StateQ$, $F$ is a database of facts
 against which we issue the query, $F'$ is a partial answer (i.e., an answer built so far in the rewriting process), and $S$ is a stack of sort $\StackQ$ which simulates structural recursion.
 Normal forms encapsulating an answer to $Q$ are constructed with
$\Ans:\FactSet\rightarrow\StateQ$.
Terms of sort $\StackQ$ are constructed from local computation frames of sort
$\NodeQ$, where $\NodeQ<\StackQ$, using an associative
binary operator  $\_\_:\StackQ\;\StackQ\rightarrow\StackQ$ with identity element $\Empty$.
Constructors
of frames are indexed by sub-queries $R$ of $Q$,
 and lists of distinct variable names
$\vec{v}:=v_1,\ldots, v_n$  of respective sorts $s_1,\ldots, s_n$
($\vec{v}$ can be of any length
as long as $\{\vec{v}\}\subseteq\Var(Q)$ and
it contains all free variables of $R$):
\begin{gather*}
[\_,\ldots,\_]^{\vec{v}|q}_R:s_1\ldots s_n\rightarrow\NodeQ,\quad
[\_,\ldots,\_|\_]^{\vec{v}|q}_{R}:s_1\ldots s_n\;\StackC\rightarrow\NodeQ,\\
[\_|\_,\ldots,\_]^{\vec{v}|q}_{R}:\FactSet\;s_1\ldots s_n\rightarrow\NodeQ\ \text{if}\ R=\From P\mathbin{.}R'.
\end{gather*}
As $\vec{v}$ can be empty, the above signature templates include
$[]^{|c}_R:\rightarrow\NodeQ$, etc.
As in the case of $\QLang^\cond$, variables in super- and sub-scripts are
part of function names and are never matched or substituted --- Remark~\ref{MetaRem}  applies  here
 and in the next section.
Frames of the form $[\vec{a}]^{\vec{v}|q}_R$, $[\vec{a}|S]^{\vec{v}|q}_R$,
or $[F'|\vec{a}]^{\vec{v}|q}_R$
are called $(R, \sigma)$-frames, where $\sigma:=\{\vec{a}/\vec{v}\}$
 is the current substitution.
They indicate evaluation of $\sigma(R)$.
Conditional frames  $[\vec{a}|S]^{\vec{v}|q}_R$ are used in
 evaluation of conditionals $\phi\Rightarrow R$, where
$S:\StackC$ represents evaluation of $\phi$.
Iterator frames $[F'|\vec{a}]^{\vec{v}|q}_{\From P\mathbin{.}R}$,
represent iterative evaluation of $\sigma(\From P\mathbin{.}R)$.
Multiset $F'$, called iterator state, contains facts available for matching with
$P_!\circ P_?$. Given a database $F$,
to evaluate a closed query $Q$ we rewrite a state term
$\Init_Q(F):=\{F, \Empty, []^{|q}_Q\}^q$
until a normal form $\Ans(F')$ is reached.
Then we conclude that
 evaluation of $Q$ on $F$ yields $F'$ as an answer.

 Now we are ready to define the rules of $\cR_{\Sigma, \cD}^\query(Q)$.
Literal  facts
are added to the partial
answer multiset after applying the current substitution, and empty queries return nothing:
 \begin{equation}
 \LFact:\bigl\{F, F', S[\vec{v}]^{\vec{v}|q}_{f}\bigr\}^q\Rightarrow
  \bigl\{F,F'\circ f,S\bigr\}^q,\quad
  \LEmpty:\bigl\{F, F', S[\vec{v}]^{\vec{v}|q}_{\Empty}\bigr\}^q\Rightarrow
  \bigl\{F,F',S\bigr\}^q.
 \end{equation}

Evaluation of ``union'' $\_\Con\_$  is implemented by replacing the frame corresponding to $R_1\Con R_2$
with two frames corresponding to $R_1$ and $R_2$, respectively:
\begin{equation}
\LUnUnfl:\ \bigl\{F, F', S[\vec{v}]^{\vec{v}|q}_{R_1\Con R_2}\bigr\}^q\Rightarrow
 \bigl\{F,F', S[\vec{v}]^{\vec{v}|q}_{R_2}[\vec{v}]^{\vec{v}|q}_{R_1}\bigr\}^q
\end{equation}

To compute a conditional $\phi\Rightarrow R$ we first embed a stack representing computation of condition
$\phi$ within the
frame corresponding to the conditional. Once this condition is evaluated, we either evaluate $R$
if the condition is satisfied, or drop the conditional if it is not:
 \begin{align}
 \LCondUnfl:&\ \bigl\{F, F', S[\vec{v}]^{\vec{v}|q}_{\phi\Rightarrow R}\bigr\}^q\Rightarrow
 \bigl\{F, F', S\bigl[\vec{v}\;|\;[\vec{v}]^{\vec{v}|c}_\phi\bigr]^{\vec{v}|q}_{R}\bigr\}^q,\nonumber\\
 \LCondFldf:&\ \bigl\{F, F', S[\vec{v}\;|\;\Res(\ffalse)]^{\vec{v}|q}_{R}\bigr\}^q\Rightarrow
 \{F, F', S\}^q,\nonumber\\
 \LCondFldt:&\ \bigl\{F, F', S[\vec{v}\;|\;\Res(\ttrue)]^{\vec{v}|q}_{R}\bigr\}^q\Rightarrow
 \bigl\{F, F', S[\vec{v}]_R^{\vec{v}|q}\bigr\}^q.
 \label{CondEqQuery}
 \end{align}
 To compute  $\phi$ we add, for every
 rule $\lambda:\{F, S'\}^c\Rightarrow\{F,S''\}^c\ \text{if}\ C$ in $\cR_{\Sigma, \cD}^\cond(\phi)$
the rule schema
\begin{equation}
\lambda^q:\{F,F',S[\vec{v}\;|\;S']^{\vec{v}|q}_R\}^q\Rightarrow\{F,F', S[\vec{v}\;|\;S'']^{\vec{v}|q}_R\}^q\ \text{if}\ C
\label{CondForAllEq}
\end{equation}
Evaluation of $\From\_.\_$ subquery is initialized with the whole database available for matching:
\begin{equation}
\FromInit:\bigl\{F, F', S[\vec{v}]^{\vec{v}|q}_{\From P\mathbin{.}R}\bigr\}^q
\Rightarrow
\bigl\{F, F', S[F\;|\; \vec{v}]^{\vec{v}|q}_{\From P\mathbin{.}R}\bigr\}^q
\end{equation}
Let  $\vec{w}$ be a sequence of
 all the distinct variables in $\Var(P)\setminus\{\vec{v}\}$. Let $\sigma$ be the current substitution.
Rule $\FromUnfl$ pushes onto
the stack a $(\sigma',R)$-frame, where $\sigma'=\sigma\cup\{\vec{b}/\vec{w}\}$
is defined by matching $F''\circ\sigma(P_!\circ P_?)$ with iterator state, and it removes $\sigma'(P_?)$
from the iterator state:
 \begin{equation}
 \FromUnfl:
 \bigl\{F,F',S\bigl[F''\circ P_!\circ P_{?}\;|\;
 \vec{v}\bigr]^{\vec{v}|q}_{\From P\mathbin{.}R}\bigr\}^c
 \Rightarrow
  \bigl\{F,F',S\bigl[F''\circ P_!\;|\; \vec{v}\bigr]^{\vec{v}|q}_{\From P\mathbin{.}R}
 [\vec{v},\vec{w}]^{\vec{v},\vec{w}|q}_{R}\bigr\}^q
 \end{equation}
 We keep applying $\FromUnfl$ until we cannot match
$F''\circ\sigma(P_!\circ P_?)$ with iterator state. Then we remove
the iterator frame from the stack. To prevent
premature application, rule schema $\FromEnd$ is conditional, where the condition uses
functions $\mu_{P, \vec{v}}:\TPPattern\;s_1\ldots s_n\rightarrow\Succ$ defined for each
$\vec{v}\subseteq\Var(Q)$ and
pattern $P$ occurring in $Q$ with the single equation
$ \mu_{P,\vec{v}}(F\circ P_!\circ P_?,\vec{v})=\yes$ (cf. Equation~\eqref{EndFuncDefEq}):
\begin{equation}
\FromEnd:\{F, F', S[F''\;|\;\vec{v}]^{\vec{v}|q}_{\From P\mathbin{.}R}\}^q
\Rightarrow\{F,F',S\}^q\quad\text{if}\ (\mu_{P,\vec{v}}(F'', \vec{v})=\yes)=\ffalse.
\end{equation}

Finally, the rule $\LAns:\{F,F',\Empty\}^c\Rightarrow\Ans(F')$
finishes evaluation of $Q$.

The following result can be proven similarly to Theorem~\ref{CondTerminating}:
\begin{theorem}
$\cR_{\Sigma, \cD}^\query(R)$ is a terminating rewriting system.
\end{theorem}

The following useful observation can be trivially verified by examining the rule schemas:
\begin{lemma}
\label{SubLemmaQuery}
Let $R$ be a subcondition of $Q$. Then, for all multisets of facts $F$, $F'$, $F''$, stacks $S$,
 lists of variables $\vec{v}=v_1,\ldots,v_n$ and values $\vec{a}=a_1, \ldots,a_n$,
 $\{F, F',S[\vec{a}]^{\vec{v}|q}_R\}^q\rightarrow^*\{F, F'\circ F'', S\}^q$
in $\cR_{\Sigma, \cD}^\query(Q)$
if and only if
$\{F,\Empty,[]^{|q}_{\sigma(R)}\}^q\rightarrow^!\Ans(F'')$ in $\cR_{\Sigma, \cD}^\query(\sigma(R))$,
where substitution $\sigma:=\{\vec{a}/\vec{v}\}$.
\end{lemma}

The following example shows that evaluation of queries in $\QLang^\query$ does not,
in general, return a unique answer, and,
in particular, that $\cR_{\Sigma, \cD}^\query(Q)$ for general queries $Q$ is not confluent.
\begin{example}
Let $\Token:\rightarrow\Fact$ and $b:\Nat\rightarrow \Fact$.
Consider the query $Q:=\From \Rett{\Token}\circ \Keep{b(x)}\mathbin{.}b(x)$ executed against
database $F:=\Token\circ b(1)\circ b(2)$. Then $\Init_Q(F)$ can normalized in two ways:
\begin{equation*}
\Init_Q(F)\xrightarrow{\FromInit}
\bigl\{F, \Empty, [F\;|\;]^{|q}_{Q}\bigr\}^q
\xrightarrow{\FromUnfl}
\left\{
\begin{array}{l}
\bigl\{F, \Empty, [b(1)\circ b(2)\;|\;]^{|q}_{Q}[1]^{x|q}_{b(x)}\bigr\}^q
\rightarrow^!\Ans(b(1))\\
\bigl\{F, \Empty, [b(1)\circ b(2)\;|\;]^{|q}_{Q}[2]^{x|q}_{b(x)}\bigr\}^q
\rightarrow^!\Ans(b(2))
\end{array}
\right..
\end{equation*}
Queries  like $Q$ are useful as a
simulation of a non-deterministic choice (say, by a human agent) of a subset of
values stored in the database with a fixed maximal cardinality.
E.g., $\From \Rett{b(x)}\mathbin{.}b(x)$ returns all ``$b$-facts''
stored in the database.
Query $Q$ defined above, however, chooses (with repetitions)
at most as many $b$-facts  as there are tokens $\Token$. Query
$\From \Rett{\Token\circ b(x)}\mathbin{.}b(x)$ avoids repetitions.
\end{example}

The following is the bisimulation-like definition of
logical equivalence between queries in $\QLang^\query$:
\begin{definition}
Let $Q_1$ and $Q_2$ be two conditions in $\QLang^\query$. Recall that
$\Init_Q(F):=\{F,\Empty,[]^{|q}_Q\}^q$.
We say that $Q_1$ is logically equivalent to $Q_2$, writing $Q_1\equiv Q_2$, if and only if,
for all ground multisets of facts $F$ and $F'$, and ground substitutions $\sigma$ such that $\sigma(Q_1)$
and $\sigma(Q_2)$ are closed,
we have
\begin{equation*}
\Init_{\sigma(Q_1)}(F)\rightarrow^!\Ans(F')\quad\text{iff}\quad \Init_{\sigma(Q_2)}(F)\rightarrow^!\Ans(F').
\end{equation*}
\end{definition}

In other words, queries are equivalent if they can match each other's answers.
The following result is an immediate consequence of Lemma~\ref{SubLemmaQuery}:
\begin{lemma}
Logical equivalence on queries in $\QLang^\query$ is an equivalence relation and a congruence, i.e., if
$\kappa$ is a position in a query $Q$ in $\QLang^\query$ such that $Q|_\kappa$ is
a query, and $R\equiv Q|_\kappa$, then $Q\equiv Q[R]_\kappa$.
\end{lemma}

We leave  proof of the next observation to the reader:

\begin{lemma}
\label{RenameQuery}
For any closed query $Q$ in $\QLang^\query$, and any renaming $\sigma$,
$Q\equiv\sigma(Q)$.
\end{lemma}

The following clarification of semantics of queries in $\QLang^\query$ is proven
similarly to Lemma~\ref{LemmaClearSem}:
\begin{lemma}
\label{LemmaClearSemQuery}
For all ground multisets of facts $F$, $F'$ and $F''$, all closed queries $Q$, $Q_1$,
$Q_2$, and $\From  P\mathbin{.}R$
in $\QLang^\query$, and all closed conditions $\phi$ in $\QLang^\cond$,
the following statements hold:
\begin{enumerate}
\item
If $\Init_Q(F)\rightarrow^!\Gamma$ then $\Gamma=\Ans(G)$ for some ground multiset of facts $G$.
\item Let $f$ be a fact.
If $\Init_{f}(F)\rightarrow^!\Gamma$ then $\Gamma=\Ans(f)$. If
$\Init_{\Empty}(F)\rightarrow^!\Gamma$ then $\Gamma=\Ans(\Empty)$.
\eject
\item Let $F'\neq\Empty$. In this case
$\Init_{\phi\Rightarrow Q}(F)\rightarrow^!\Ans(F')$ iff
$\Init_\phi(F)\rightarrow^!\Sat(\ttrue)$ and $\Init_Q(F)\rightarrow^!\Ans(F')$.
\item $\Init_{\phi\Rightarrow Q}(F)\rightarrow^!\Ans(\Empty)$ iff
either (non-exclusively) $\Init_\phi(F)\rightarrow^!\Sat(\ffalse)$ or
$\Init_Q(F)\rightarrow^!\Ans(\Empty)$.
\item $\Init_{Q_1\Con Q_2}(F)\rightarrow^!\Ans(F')$ iff
$\Init_{Q_1}(F)\rightarrow^!\Ans(F_1)$ and
$\Init_{Q_2}(F)\rightarrow^!\Ans(F_2)$ for some multisets $F_1$ and $F_2$ such that
$F'=F_1\circ F_2$.
\item $\Init_{\From P\mathbin{.}R}(F)\rightarrow^!\Ans(F')$ iff
there exist lists of ground
multisets of facts $F_0,\ldots,F_n$, $G_0,\ldots, G_{n-1}$,
and $H_0, \ldots, H_{n-1}$,
 and a sequence
of substitutions $\sigma_0,\sigma_1,\ldots,\sigma_{n-1}$ such that
\begin{enumerate}
\item $F_0=F$,
$F_{i+1}=G_i\circ \sigma_i(P_!)$, and $F_i=G_i\circ\sigma_i(P_!\circ P_{?})$,
for all $i\in\{0,\ldots, n-1\}$,
\item
$\Init_{\sigma_i(R)}(F)\rightarrow^!\Ans(H_i)$, for all $i\in\{0,\ldots, n-1\}$,
\item there exists no substitution $\sigma_n$ and multiset of facts $G_n$ such that
$F_n=G_n\circ\sigma_n(P_!\circ P_{?})$,
\item $F'=H_0\circ H_1\circ\cdots H_{n-1}$.
\end{enumerate}
\end{enumerate}
\end{lemma}

\begin{lemma}
\label{EquivLemmaQueries}
The following logical equivalences hold between conditions in $\QLang^\query$:
\begin{gather*}
\Empty\Con Q\equiv Q,\quad
Q_1\Con Q_2\equiv Q_2\Con Q_1,\quad
Q_1\Con(Q_2\Con Q_3)\equiv (Q_1\Con Q_2)\Con Q_3,\\
\False\Rightarrow R\equiv \Empty,\quad
\neg\False\Rightarrow R\equiv R,\quad
\exists P\mathbin{.}\Empty\equiv \Empty
\end{gather*}
\end{lemma}

The next results show that non-confluence of queries
makes some natural equivalences invalid:
\begin{lemma}
Let $\phi$ be a condition in $\QLang^\cond$ and let $Q_1$, $Q_2$ be queries in $\QLang^\query$. For all ground multisets of facts $F$, $F'$ and all substitutions $\sigma$
such that $\sigma(\phi)$, $\sigma(Q_1)$ and $\sigma(Q_2)$ are closed, we have
\begin{equation}
\label{ImplOneDirCond}
\Init_{\sigma((\phi\Rightarrow Q_1)\Con(\phi\Rightarrow Q_2))}(F)\rightarrow^!\Ans(F')\quad
\text{if}\quad
\Init_{\sigma(\phi\Rightarrow (Q_1\Con Q_2))}(F)\rightarrow^!\Ans(F'),
\end{equation}
however, the inverse implication does not hold in general. If $\phi$ is deterministic
(Definition~\ref{NonDeterministicCondDef}), then
\begin{equation}
\label{EquivCond}
(\phi\Rightarrow Q_1)\Con(\phi\Rightarrow Q_2)\equiv \phi\Rightarrow (Q_1\Con Q_2).
\end{equation}
\end{lemma}
\begin{proof}
To prove the implication~\eqref{ImplOneDirCond} assume $\Init_{\sigma(\phi\Rightarrow (Q_1\Con Q_2))}(F)\rightarrow^!\Ans(F')$. Either $F'\neq\Empty$ or $F'=\Empty$. In the first case, by
Lemma~\ref{LemmaClearSemQuery}, p.~3 , our assumption is equivalent to
$\Init_\phi(F)\rightarrow^!\Res(\ttrue)$ and
$\Init_{\sigma(Q_1\Con Q_2)}(F)\rightarrow^!\Ans(F')$, the latter of which, in turn, is equivalent, by
Lemma~\ref{LemmaClearSemQuery}, p.~5, to
$\Init_{\sigma(Q_1)}(F)\rightarrow^!\Ans(F_1)$ and
$\Init_{\sigma(Q_2)}(F)\rightarrow^!\Ans(F_2)$  for some multisets $F_1$ and $F_2$ such that
$F'=F_1\circ F_2$. But this is equivalent, by
Lemma~\ref{LemmaClearSemQuery}, p.~3 and~5, to
$\Init_{\sigma((\phi\Rightarrow Q_1)\Con(\phi\Rightarrow Q_2))}(F)\rightarrow^!\Ans(F').$
The case $F'=\Empty$ is dealt with similarly.

To show that, in general, the inverse implication does not hold, consider
$\phi:=\exists \Rett{a\circ b(x)}\mathbin{.}\{x=1\}$, $Q:=\phi\Rightarrow (c\Con d)$,
$Q':=(\phi\Rightarrow c)\Con(\phi\Rightarrow d)$,
where $a, c, d:\rightarrow\Fact$ and $b:\Nat\rightarrow\Fact$.
Let $F:=a\circ b(1)\circ b(2)$.
Since $\Init_\phi(F)\rightarrow^!\Sat(B)$ for $B\in\{\ttrue,\ffalse\}$,
$\Init_Q(F)\rightarrow^!\Ans(F')$ iff $F'\in\{\Empty, c\circ d\}$,
whereas
$\Init_{Q'}(F)\rightarrow^!\Ans(F')$ iff $F'\in\{\Empty, c, d, c\circ d\}$.

\medskip
Assume that
$\phi$ is deterministic.
By Theorem~\ref{ConfluentPartCond},  either
$\Init_{\sigma(\phi)}(F)\rightarrow^!\Sat(\ttrue)$ or
$\Init_{\sigma(\phi)}(F)\rightarrow^!\Sat(\ffalse)$, but not both.
Let $R_1:=\sigma\bigl((\phi\Rightarrow Q_1)\Con(\phi\Rightarrow Q_2)\bigr)$,
$R_2:=\sigma\bigl(\phi\Rightarrow (Q_1\Con Q_2)\bigr)$.
If $\Init_{\sigma(\phi)}(F)\rightarrow^!\Sat(\ffalse)$ then $\Init_{R_i}(F)\rightarrow\Ans(F')$
iff  $F'=\Empty$ for $i\in\{1,2\}$. If
$\Init_{\sigma(\phi)}(F)\rightarrow^!\Sat(\ttrue)$ then, for $i\in\{1,2\}$,
$\Init_{R_i}(F)\rightarrow\Ans(F')$ iff $F'=F_1\circ F_2$ for some
multisets $F_1$ and $F_2$ such that
$\Init_{Q_j}(F)\rightarrow\Ans(F_j)$, $j\in\{1,2\}$.
\end{proof}

\begin{example}
In general, equivalence
$\From P\mathbin{.}(Q_1\Con Q_2)\equiv(\From P\mathbin{.}Q_1)\Con(\From P\mathbin{.}Q_2)$
is not valid. Indeed,
let $a:\rightarrow\Fact$ and $b, c:\Nat\rightarrow\Fact$, and
let $P:=\Rett{a\circ b(x)}$,
$Q:=\From P\mathbin{.}(b(x)\Con c(x))$,
$Q':=(\From P\mathbin{.}b(x))\Con (\From P\mathbin{.}c(x))$.
Suppose that $F:=a\circ b(1)\circ b(2)$. Then
$\Init_{Q}(F)\rightarrow^!\Res(F')$ iff
$F'\in\bigl\{b(i)\circ c(i)\;|\;i\in\{1,2\}\bigr\}$. However, $\Init_{Q'}(F)\rightarrow^!\Res(F')$
iff $F'\in\bigl\{b(i)\circ c(j)\;|\;i,j\in\{1,2\}\bigr\}$.
\end{example}

Here we define a class of queries in $\QLang^\query$ which evaluate to unique answers:
\begin{definition}
A query $Q$ in  $\QLang^\query$ is called {\em deterministic} if
all quantification patterns in $Q$
(including those inside conditions) contain
only single facts with unique matching property.
\end{definition}

\begin{theorem}
\label{ConfluentPartQuery}
Let $Q$ be a deterministic query in $\QLang^\query$. Then
$\cR_{\Sigma,\cD}(Q)$ is confluent, and, in particular, given a ground multiset of facts $F$,
there is a unique multiset of facts $F'$ such that
$\Init_Q(F)\rightarrow^!\Ans(F')$.
\end{theorem}
\begin{proof}
As semiconfluence implies confluence, to show confluence at $t:\StateQ$ it
suffices to prove that if
$t'\leftarrow t\rightarrow^{+} t''$ for some $t'$ and $t''$,
then there exists $s$ such that
$t'\rightarrow^* s\;{}^*\!\!\leftarrow t''$. Semiconfluence is immediate
at irreducible terms $\Ans(F')$, as well as terms $\{F,F',\Empty\}^q$
which can only rewrite to $\Ans(F')$.
Let $t:=\{F, F', SA\}^q$
where $S:\StackQ$, and $A:\NodeQ$.
If  there exists a unique multiset of facts $K$ such that
every rewrite sequence starting with $t$
either (1) contains $h:=\{F, F'\circ K, S\}^q$ or (2) it can be extended to reach $h$, then semiconfluence
holds at $t$. Indeed, in this case, either $s=t''$ witnesses the semiconfluence (if (1) holds for $t\rightarrow^{+}t''$)
or $s=h$ does (if (2) holds for $t\rightarrow^{+}t''$). It remains to prove the existence of the unique multiset $K$.
We argue by induction on the structure of formulas indexing the frame $A$ on the top of the stack.
Most cases are dealt by trivial application of Lemmas~\ref{LemmaClearSemQuery} and~\ref{SubLemmaQuery}.
For frames related to conditionals $\_\Rightarrow\_$  we have to also use Lemma~\ref{ConfluentPartCond}.
The only non-trivial part of the proof concerns frames $A$ of the
form $[F''\;|\;\vec{a}]^{\vec{v}|q}_{\From P\mathbin{.}R}$.

Under the theorem's
assumption, $P=\Rett{f}$, where $f$ is a fact with
unique matching property (Definition~\ref{UniqMatchPropDef}).
Let $\vec{w}$ be a sequence of all variables in $\Var(f)\setminus\{\vec{v}\}$.
If there is no fact $f'$ in $F'$ such that $f'$ matches
$\{\vec{a}/\vec{v}\}(f)$
 then necessarily $K=\Empty$. Otherwise
$F'=G\circ f_1\circ \cdots f_n$ for some $n>0$, where
(1) for all $i\in\{1,\ldots, n\}$ there exists a unique substitution
$\sigma_i=\{\vec{a}/\vec{v},\vec{b^i}/\vec{w}\}$ such that $f_i=\sigma_i(f)$,
(2) there is no fact $f'$ in $G$ such that $f'$ matches
$\{\vec{a}/\vec{v}\}(f)$.
In this case, necessarily, by Lemma~\ref{LemmaClearSemQuery}, p.~6,
$K=K_1\circ \cdots\circ K_n$, where, for all $i\in\{1,\ldots, n\}$, $K_i$
is the unique multiset of facts (uniqueness and existence follows from
inductive assumption) such that $\Init_{\sigma_i(R)}(F)\rightarrow^!\Ans(K_i)$.
\end{proof}

There is a useful relationship between queries and conditions:
\begin{lemma}
Let $r:s_1\ldots s_n\rightarrow\Fact$,
let $Q$ be a query in $\QLang^\query$,
and let $\vec{x}:=x_1,\ldots,x_n$ be a list of $n$ distinct variables
such that $\{\vec{x}\}\cap\Var(Q)=\emptyset$ and
 $x_i:s_i$ for $i\in\{1,\ldots,n\}$.
 Then,  there exists a
condition $\phi_Q^{\vec{x}|r}$ in $\QLang^\cond$ such that,
for all ground multisets of facts $F$,
and for all
 ground substitutions $\sigma:=\{\vec{t}/\vec{x},\vec{a}/\vec{v}\}$
such that $\{\vec{a}/\vec{v}\}(Q)$ is closed,
$\sigma(\phi_Q^{\vec{x}|r})$ is closed, and, moreover,
\begin{gather}
\Init_{\sigma(\phi_Q^{\vec{x}|r})}(F)\rightarrow^!\Sat(\ttrue)\quad\text{iff}\quad
\exists F'\mathbin{.}\bigl(\Init_Q(F)\rightarrow^!\Ans(F'\circ r(\vec{t}))\bigr),\nonumber\\
\Init_{\sigma(\phi_Q^{\vec{x}|r})}(F)\rightarrow^!\Sat(\ffalse)\quad\text{iff}\quad
\exists F'\mathbin{.}\bigl(r(\vec{t})\notin F'\wedge\Init_Q(F)\rightarrow^!\Ans(F')\bigr),
\end{gather}
i.e., $\sigma(\phi_Q^{\vec{x}|r})$ evaluates to $\ttrue$ (resp. $\ffalse$) on $F$
if and only if  there is some $F'$ returned by $Q$ when evaluated against $F$, such that
$r(\vec{t})\in F'$ (resp. $r(\vec{t})\notin F'$).
Moreover, if $Q$ is deterministic, then so is $\phi_Q^{\vec{x}|r}$
\label{IffCondQueryRel}
\end{lemma}
\begin{proof}
We define $\phi_Q^{\vec{x}|r}$ by recursion on the structure of a query $Q$:
\begin{gather*}
\phi_{\Empty}^{\vec{x}|r}=\False,\
\phi_{h(\vec{t})}^{\vec{x}|r}=\{r(\vec{x})=h(\vec{t})\},\
\phi_{Q_1 \Con Q_2}^{\vec{x}|r}=\phi_{Q_1}^{\vec{x}|r}\vee\phi_{Q_2}^{\vec{x}|r},\
\phi_{\psi\Rightarrow R}^{\vec{x}|r}=\psi\wedge\phi_{R}^{\vec{x}|r},\
\phi_{\From P\mathbin{.}R}^{\vec{x}|r}=\exists P.\phi_{R}^{\vec{x}|r}.
\end{gather*}
The easy if laborious proof that $\phi_Q^{\vec{x}|r}$ really satisfies all the conditions in the statement
is left to the reader.
\end{proof}

\begin{example}
Consider the following query in $\QLang^\query$ (where $x$ and $y$ are distinct variables):
\begin{equation*}
Q:=\From \Ret{f(x, y)}\mathbin{.}\bigl((\{x=y\}\Rightarrow r(x))\Con r(s(x))\bigr).
\end{equation*}
Thus, for any fact of the form $f(x, y)$ in the database, $Q$ will output $r(s(x))$, and, if $x=y$, also $r(x)$. Using the recursive formula from the proof of Lemma~\ref{IffCondQueryRel} we easily see that
\begin{equation*}
\phi^{z|r}_Q:=\exists\Ret{f(x, y)}\mathbin{.}\bigl((\{x=y\}\wedge \{r(z)=r(x)\})
\vee\{r(z)=r(s(x))\}\bigr).
\end{equation*}
\end{example}

The next result shows that $\QLang^\query$ can emulate relational algebra.
\begin{theorem}
\label{RelAlgTheorem}
Denote by $\text{RelAlg}(\cS, \cD)$ the relational algebra over the relational schema $\cS$ and the domain $\cD$ of atomic values (which we silently identify with its algebraic representation, where types are sorts and
predicates are represented with equationally defined operators
into the $\Bool$ sort). Furthermore, assume that for each relational
algebra expression $R$ in $\text{RelAlg}(\cS, \cD)$   a function symbol
$\cR_R:s_1\ldots s_n\rightarrow \Fact$ is in $\Sigma_F$, where $s_i$ is the sort (domain) of
the $i$-th column of $R$ and $n$ is the arity of $R$.
For any relational database $I$ with schema $\cS$, let $\Trb(I)$ be a multiset corresponding to the database $I$. More precisely,
\begin{equation*}
\Trb(I):=\circ\{\cR_r(\vec{t})\;|\;\text{$r$ is a relation symbol in $\cS$ and\ }
(\vec{t})\in r^I\},
\end{equation*}
where $r^I$ is the set of tuples of $r$ in $I$.
Then, for all  formulas $R$ in $\text{RelAlg}(\cS, \cD)$,
there exists a closed query $\Tr(R)$ in $\QLang^\query$ such that
for all relational databases $I$ with schema $\cS$,
\begin{equation}
\label{ConditionTrans}
\exists F\mathbin{.}\bigl(\Init_{\Tr(R)}(\Trb(I))\rightarrow^!\Ans(F\circ \cR_R(\vec{t}))\bigr)
\quad \text{iff}\quad
(\vec{t})\in \Eval_I(R),
\end{equation}
where  $\Eval_I(R)$ is the set of tuples obtained by
evaluation of relational query $R$ against the relational database $I$. Furthermore,
for any relational expression $R$, $\Tr(R)$ is deterministic.
\end{theorem}
\begin{proof}
To prove the theorem we  define $\Tr(R)$ by recursion on the structure
of $R$. We consider relational algebra expressions to be constructed with base relations, projections, selections, set unions, Cartesian products, and set differences. Other well known relational operators such as joins can be defined in terms of basic operators mentioned above.
Attribute renaming is not relevant in our case, since we represent relations with positional arguments only. We denote set unions, Cartesian products, and set differences with the usual mathematical notation (e.g., $R\cup S$, $R\times S$ and $R\setminus S$, respectively).
Notation for projections and selections is less standardised (and needs to be adapted to relations with positional arguments). We denote by $\pi_{i_1,\ldots,i_k}(R)$ the projection of $R$ onto $i_1$'th, $i_2$'th, $\ldots$, and $i_k$'th ``leg'' (i.e., on a single tuple,
$\pi_{i_1,\ldots,i_k}((t_1, t_2,\ldots, t_n)):=(t_{i_1}, t_{i_2},\ldots,t_{i_k})$).
We denote by $\sigma_\phi(R)$ the selection with condition $\phi$ applied to $R$ (i.e., it returns those tuples of $R$ which satisfy $\phi$). If $R$ is $n$-ary, then we will assume that
in $\phi$ expression $\$_i$ corresponds to the $i$-th column of $R$, for $i\in\{1, \ldots, n\}$.

\medskip
Let $r$ be a base relation of
arity $k$ (for simplicity we omit the typing information), and let
 $\vec{x}$ be a list of $k$ distinct variables. Then
$
\Tr(r):=\From  \Ret{\cR_r(\vec{x})}\mathbin{.}\cR_r(\vec{x}).
$
We now consider selected relational operators:
\begin{enumerate}
\item Let $R$  be a relational formula defining an $n$-ary relation,
and let $i_1,\ldots,i_k$ be a subsequence of $1,\ldots,n$.
We define $\Tr(\pi_{i_1,\ldots,i_k}(R))$ to be  $\Tr(R)$ with
each subquery  of the form $\cR_{R}(t_1, \ldots, t_n)$
replaced with $\cR_{\pi_{i_1,\ldots,i_k}(R)}(t_{i_1},\ldots,t_{i_k})$.
\item Let $R$ be a relational formula of arity $k$, and let
$\phi$ be a term of sort $\sort{Bool}$ representing condition on rows of $R$
(where in $\phi$ special variable $\$_i$, $i\in\{1, \ldots, n\}$ corresponds to the
$i$-th ``attribute'' of $R$).
We define $\Tr(\sigma_\phi(R))$ to be $\Tr(R)$ with  each subquery of the form
$\cR_R(\vec{t})$ replaced with
 $\bigl\{\{\vec{t}/\vec{\$}\}(\phi)\bigr\}\Rightarrow\cR_{\sigma_\phi(R)}(\vec{t})$.
 \item $\Tr(R_1\cup R_2):=\Tr(R_1)\Con\Tr(R_2)$.
\item Let $R$ and $S$ be relational formulas.
Let $\Tr(S)':=\sigma(\Tr(S))$ for some renaming $\sigma$ such that
$\Var(\Tr(R))\cap\Var(\Tr(S)')=\emptyset$ ($\Tr(S)\equiv\Tr(S)'$ by
Lemma~\ref{RenameQuery} since  $\Tr(S)$ is closed).
Further, let $\alpha(\vec{t})$ be  $\Tr(S)'$
with each subquery of the form
$\cR_S(\vec{s})$ replaced with
$\cR_{R\times S}(\vec{t}, \vec{s})$. Then
we define $\Tr(R\times S)$ to be
$\Tr(R)$ with each subquery of the form  $\cR_R(\vec{t})$
replaced with
$\alpha(\vec{t})$.
\item Let $R$ and $S$ be relational formulas
of arity $k$. Let $\Tr(S)'$ be like in the previous point.
Let $\vec{x}$ be a list of $k$ distinct variables such that
$\{\vec{x}\}\cap(\Var(\Tr(R))\cup\Var(\Tr(S)'))=\emptyset$.
We define $\Tr(R\setminus S))$ to be
$\Tr(R)$ with each  subquery of the form
$\cR_R(\vec{t})$  replaced with
$\neg\{\vec{t}/\vec{x}\}\bigl(\phi_{\Tr(S)'}^{\vec{x}|\cR_S}\bigr)\Rightarrow\cR_{R\setminus S}(\vec{t})$ (see Lemma~\ref{IffCondQueryRel}).
\end{enumerate}
An easy induction on the structure of $R$ shows that $\Tr(R)$ is closed and deterministic, and
Equation~\eqref{ConditionTrans} is satisfied (in
 the case of induction step for
  set difference we also use Lemma~\ref{IffCondQueryRel})
\end{proof}

\begin{remark}
Relational queries $R$ evaluate to sets of tuples.
However, $\Tr(R)$ may evaluate to a multiset of facts --- e.g., when evaluating unions
duplicate facts are not removed.
\end{remark}

\begin{example}
Let $r$ and $s$ be binary relations. Consider the following relational algebra expression:
\begin{equation*}
R:=\pi_{1, 4}\bigl(\sigma_{\$_2=\$_3}(r\times s)\bigr).
\end{equation*}
Thus, $(x,y)\in R$ iff $(x, z)\in r$ and $(z,y)\in s$ for some $z$. Let represent $R$ as a query in $\QLang^\query$ using definition of $\Tr(\_)$ from the proof of Theorem~\ref{RelAlgTheorem}.
First,
\begin{equation*}
\Tr(r)=\From\Ret{\cR_r(x_1, x_2)}\mathbin{.}\cR_r(x_1, x_2),\quad
\Tr(s)=\From\Ret{\cR_s(y_1, y_2)}\mathbin{.}\cR_s(y_1, y_2),
\end{equation*}
where $x_1, x_2, y_2, y_2$ are distinct variables. Then, by the point~4 in the proof
\begin{equation*}
\Tr(r\times s)=\From \Ret{\cR_r(x_1, x_2)}\mathbin{.}\From\Ret{\cR_s(y_1, y_2)}\mathbin{.}
\cR_{r\times s}(x_1, x_2, y_1, y_2).
\end{equation*}
Finally, to deal with selection and projection we apply points~2 and 1, respectively, from the proof:
\begin{multline*}
\Tr\bigl(\pi_{1, 4}\bigl(\sigma_{\$_2=\$_3}(r\times s)\bigr)\bigr)\\
=
\From \Ret{\cR_r(x_1, x_2)}\mathbin{.}\From\Ret{\cR_s(y_1, y_2)}\mathbin{.}
\bigl(
\{x_2=y_1\}\Rightarrow
\cR_{\pi_{1, 4}(\sigma_{\$_2=\$_3}(r\times s))}(x_1, y_2)
\bigr).
\end{multline*}
\end{example}

\section{Rewriting semantics of $\QLang^\dml$}

We associate with  a DML query $Q$ in $\QLang^\dml$
the rewriting system $\cR_{\Sigma, \cD}^\dml(Q)$.
Terms rewritten with the rules of  $\cR_{\Sigma, \cD}^\dml(Q)$
are of the form $\{F,F', F_\FrMark,S\}^d$, where
 $\{\_,\_,\_,\_\}^d:\FactSet\;\FactSet\;\FactSet^\FrMark\;\StackD\rightarrow\StateD$. $F$ is the  database of facts
 against which we issue the DML query. $F$  changes during execution of the query while
 the facts are removed from it. A multiset
$F'$, expanded during query execution, contains new facts to be added to the database.
$F_\FrMark$ is a multiset of fresh facts (see Section~\ref{MofFandPSect}) from which fresh values are drawn.
  $S$ is a stack of sort $\StackD$ which simulates structural recursion.
   Normal forms encapsulating  a new database
and new multiset of fresh facts after successful or, respectively,
 failing execution of a DML query, are constructed with
$\New, \Fail:\FactSet\;\FactSet^\FrMark\rightarrow\StateD$.
We consider only terms $t$ of sort $\StateD$ which satisfy the following {\em freshness condition}:

\begin{definition}
\label{FreshCondRem}
A term $t$ of sort $\StateD$ satisfies the {\em freshness condition} if and only if $m>n$ for all positions
$\kappa$, $\kappa'$ in $t$ such that $\kappa\neq\kappa'1$,
$t|_{\kappa}=C_s(\imath^{s}_m)$ and
$t|_{\kappa'}=\imath^{s}_n$.
\end{definition}

\begin{example}
Consider the following term of $\StateD$ sort:
\begin{equation*}
\New\bigl(f(\imath^{s}_{10})\circ f(\imath^{s}_3), C_s(\imath^{s}_7)\bigr).
\end{equation*}
It does not satisfy the freshness condition having as subterms both $C_s(\imath^{s}_7)$ and
$\imath^{s}_{10}$ with $10> 7$. Our general assumption which justifies the freshness condition is that a fresh fact of the form $C_s(\imath^{s}_m)$ means that we never before used any value of the form $\imath^s_n$ with $n>m$. Clearly, terms which do not satisfy freshness condition (like the term above) violate this assumption.
\end{example}

Terms of sort $\StackD$ are constructed from local computation frames of sort
$\NodeD$, where $\NodeD<\StackD$, using an associative
binary operator  $\_\_:\StackD\;\StackD\rightarrow\StackD$ with identity  $\Empty$.
Most constructors
of frames are indexed by DML sub-queries $R$ of $Q$,
 and lists of distinct variable names
$\vec{v}:=v_1,\ldots, v_n$  of respective sorts $s_1,\ldots, s_n$
such that $\{\vec{v}\}\subseteq\Var(Q)$  contains all free variables of $R$:
\begin{gather*}
[\_,\ldots,\_]^{\vec{v}|d}_R, [\_,\ldots,\_]^{\vec{v}|d,\downarrow}_R:s_1\ldots s_n\rightarrow\NodeD,\quad
[\_,\ldots,\_\;|\;\_]^{\vec{v}|d}_{R}:s_1\ldots s_n\;\StackC\rightarrow\NodeD,\\
\Ok:\rightarrow\NodeD,\quad
[\_\;|\;\_,\ldots,\_\;|\;\_]^{\vec{v}|d}_{\From P\mathbin{.}R}:\FactSet\;s_1\ldots s_n\;\Bool\rightarrow\NodeD,
\\
[\_\;|\;\_,\ldots,\_\;|\;\_,\_]^{\vec{v}|d}_{\From P\mathbin{.}R}:\FactSet\;s_1\ldots s_n\;\Bool\;\FactSet\rightarrow\NodeD.
\end{gather*}
As $\vec{v}$ can be empty, the above signature templates include
$[]^{\vec{v}|d}_R, []_R^{\vec{v}|d,\downarrow}:\rightarrow\NodeD$, etc.
 A constant $\Ok$ marks successful branches of computation, i.e., those which created
either new facts or  $\dummy:\DQuery$. Marking such branches is necessary as facts deleted by
 unsuccessful branches are restored as soon as the branch  finishes.
Frames of the form $[\vec{a}]^{\vec{v}|d}_R$,
$[\vec{a}]^{\vec{v}|d,\downarrow}_R$,
$[\vec{a}|S]^{\vec{v}|d}_R$,
 $[F'|\vec{a}|B]^{\vec{v}|d}_R$, or
$[F'|\vec{a}|B, F'']^{\vec{v}|d}_R$
are called $(R, \sigma)$-frames, where $\sigma:=\{\vec{a}/\vec{v}\}$
 is the current substitution.
They are related to evaluation of $\sigma(R)$.
Marked frames $[\vec{a}]^{\vec{v}|d,\downarrow}$ occur in the execution of ``unions'' $\_\Con\_$.
Conditional frames  $[\vec{a}|S]^{\vec{v}|d}_R$ are used in
execution of conditionals $\phi\Rightarrow R$, where
$S:\StackC$ represents evaluation of $\phi$.
Iterator frames $[F'|\vec{a}|B]^{\vec{v}|d}_{\From P\mathbin{.}R}$ and
$[F'|\vec{a}|B,F'']^{\vec{v}|d}_{\From P\mathbin{.}R}$
represent iterative  execution of  $\sigma(\From P\mathbin{.}R)$.
Multiset $F'$, called iterator state,
contains facts available for matching with
$P_!\circ P_?\circ P_0$ (cf.~Remark~\ref{PatternNotationRemark}).
Iteration status $B$ is equal to $\ttrue$ iff  the iteration
already generated either new facts or $\dummy$ (i.e., if the branch related to
$\sigma(\From P\mathbin{.}R)$ was successful). ``Tentative'' iterator frames
$[F'|\vec{a}|B,F'']^{\vec{v}|q}_{\From P\mathbin{.}R}$ store multiset
$F''$ of facts deleted from the database in the present step, so that they can be restored
if the step is unsuccessful.
Given a database $F$,
in order to execute a closed DML query $Q$ we rewrite a state term
$\Init_Q(F, F_\FrMark):=\{F, \Empty, F_\FrMark, []^{|d}_Q\}^d$
until a normal form
 $\New(F', F_\FrMark^{'})$ or $\Fail(F', F_\FrMark^{'})$ is reached, indicating that a successful or, respectively, unsuccessful execution of
 $Q$ in the database $F$ yielded a new database $F'$,
 and a new multiset of fresh facts $F_\FrMark^{'}$.
 Now we are ready to define rule schemas of $\cR_{\Sigma, \cD}^\dml(Q)$.

\medskip
Two consecutive occurrences of $\Ok$ are collapsed, and
$(\Ok, \sigma)$-frames are replaced with $\Ok$:
 \begin{equation}
 \Collapse:\{F, F', F_\FrMark, S\Ok\Ok\bigr\}^d
 \Rightarrow\!\bigl\{F, F', F_\FrMark, S\Ok\}^d,\;\;
 \LDummy:\{F, F', F_\FrMark, S[\vec{v}]^{\vec{v}|d}_\dummy\}^d
  \Rightarrow\!\{F, F', F_\FrMark, S\Ok\}^d.
 \label{CollapseEq}
 \end{equation}

An $(\Empty,\sigma)$-frame
   is removed from the stack. An $(f,\sigma)$-frame, where $f$ is a fact and $\sigma$ is the
   current substitution, is replaced by $\Ok$ and $\sigma(f)$  is added to $F'$
   (since $\Var(f)\subseteq\{\vec{v}\}$, $\sigma(f)$ is closed):
  \begin{equation}
  \LEmpty:\bigl\{F,F'\!, F_\FrMark, S[\vec{v}]^{\vec{v}|d}_{\Empty}\bigr\}^d
  \Rightarrow\bigl\{F, F'\!, F_\FrMark, S\bigr\}^q,\
  \LFact:\bigl\{F,F'\!, F_\FrMark, S[\vec{v}]^{\vec{v}|d}_{f}\bigr\}^d
  \Rightarrow\bigl\{F, F'\circ f, F_\FrMark, S\Ok\bigr\}^q.
  \end{equation}

   An $(R_1\Con R_2,\sigma)$-frame is split into the $(R_1,\sigma)$-frame and the $(R_2,\sigma)$-frame.
   The $(R_2,\sigma)$-frame is marked with $\downarrow$ so that the evaluation of
 $\sigma(R_1\Con R_2)$ can be marked as successful when at least one of the branches is successful. When both branches are successful, this can produce
 two consecutive $\Ok$ constants on the stack which are then collapsed using $\Collapse$ rule in
 Equation~\eqref{CollapseEq}.
  \begin{align}
  \LUnUnfl:&\ \bigl\{F,F',F_\FrMark,S[\vec{v}]^{\vec{v}|d}_{R_1\Con R_2}\bigr\}^d
 \Rightarrow\bigl\{F,F',F_\FrMark,S[\vec{v}]^{\vec{v}|d,\downarrow}_{R_2}[\vec{v}]^{\vec{v}|d}_{R_1}\bigr\}^d,\nonumber\\
 \LUnFldOk:&\ \bigl\{F,F',F_\FrMark,S[\vec{v}]^{\vec{v}|d,\downarrow}_{R_2}\Ok\bigr\}^d
 \Rightarrow
 \bigl\{F,F',F_\FrMark,S\Ok[\vec{v}]^{\vec{v}|d}_{R_2}\bigr\}^d,\nonumber\\
 \LUnFldEmpty:&\ \bigl\{F,F',F_\FrMark,S[\vec{v}]^{\vec{v}|d,\downarrow}_{R_2}\bigr\}^d
 \Rightarrow
 \bigl\{F,F',F_\FrMark,S[\vec{v}]^{\vec{v}|d}_{R_2}\bigr\}^d,
\end{align}

 Rule schemas for execution of conditionals   $\phi\Rightarrow Q$ are similar to those in Equations~\eqref{CondEqQuery}, \eqref{CondForAllEq}:
 \begin{align}
  \LCondUnfl:&\bigl\{F, F', F_\FrMark, S[\vec{v}]^{\vec{v}|d}_{\phi\Rightarrow R}\bigr\}^d\Rightarrow
 \bigl\{F, F', F_\FrMark, S\bigl[\vec{v}\;|\;[\vec{v}]^{\vec{v}|c}_\phi\bigr]^{\vec{v}|d}_{R}\bigr\}^d,\nonumber\\[-0.5pt]
 \LCondFldf:&\bigl\{F, F', F_\FrMark, S[\vec{v}\;|\;\Res(\ffalse)]^{\vec{v}|d}_{R}\bigr\}^d\Rightarrow
 \{F, F', F_\FrMark, S\}^d,\nonumber\\[-0.5pt]
 \LCondFldt:&\bigl\{F, F', F_\FrMark, S[\vec{v}\;|\;\Res(\ttrue)]^{\vec{v}|d}_{R}\bigr\}^d\Rightarrow
 \bigl\{F, F', F_\FrMark, S[\vec{v}]_R^{\vec{v}|d}\bigr\}^d,\nonumber\\[-0.5pt]
 \lambda^d:&\{F,F', F_\FrMark, S[\vec{v}\;|\;S']^{\vec{v}|d}_R\}^d
\Rightarrow\{F,F', F_\FrMark, S[\vec{v}\;|\;S'']^{\vec{v}|d}_R\}^d\ \text{if}\ C\nonumber\\[-0.5pt]
&\quad\quad\text{for all}\ \lambda:\{F, S'\}^c\Rightarrow\{F,S''\}^c\ \text{if}\ C\  \text{in}\  \cR_{\Sigma, \cD}^\cond(\phi).
 \end{align}

Evaluation of a  $\From\_.\_$ subquery is initialized with the whole database available for
matching. Iteration status  is $\ffalse$ since evaluation is not successful yet.
\begin{equation}
\FromInit:\quad \bigl\{F, F', F_\FrMark, S[\vec{v}]^{\vec{v}|d}_{\From P\mathbin{.}R}\bigr\}^d
\Rightarrow
\bigl\{F, F', F_\FrMark, S[F\;|\;\vec{v}\;|\;\ffalse]^{\vec{v}|d}_{\From P\mathbin{.}R}\bigr\}^d.
\end{equation}

Let  $\vec{w}$ be a sequence of
 all the distinct variables in $\Var(P)\setminus\{\vec{v}\}$. Let $\sigma$ be the current substitution.
Rule $\FromUnfl$ pushes onto
the stack a $(\sigma',R)$-frame, where $\sigma'=\sigma\cup\{\vec{b}/\vec{w}\}$
is defined by matching $F''\circ\sigma(P_!\circ P_?\circ P_0)$ with iterator state and $P_\FrMark$
with the multiset of fresh facts $F_\FrMark$. It also removes
$\sigma'(P_?\circ P_0)$
from the iterator state and $\sigma'(P_0)$ from the database state, and, finally, it updates fresh facts using $\upsilon$ defined in Equation~\eqref{FreshBulkUpd}. Removed facts $\sigma'(P_0)$ are stored in the tentative iterator frame:
\begin{align}
\FromUnfl:&\quad
 \bigl\{F\circ P_!\circ P_?\circ P_0,F',  F_\FrMark\circ P_\FrMark,
 S\bigl[F''\circ P_!\circ P_?\circ P_0\;|\;
 \vec{v}\;|\;B\bigr]^{\vec{v}|d}_{\From P\mathbin{.}R}\bigr\}^d
 \nonumber\\
 &\quad
 \Rightarrow
  \bigl\{F\circ P_!\circ P_?,F',F_\FrMark\circ \upsilon(P_\FrMark),
  S\bigl[F''\circ P_!\;|\;\vec{v}\;|\;B, P_0\bigr]^{\vec{v}|d}_{\From P\mathbin{.}R}
  [\vec{v},\vec{w}]^{\vec{v},\vec{w}|d}_{R}\bigr\}^d,
\end{align}
If execution of $\sigma'(R)$ proves unsuccessful, removed facts can be restored both to iterator state and database. Otherwise, we discard them and set
 the iteration status to $\ttrue$:
\begin{align}
  \lambda^\Fld_{\From;\Empty}:&\ \bigl\{F, F', F_\FrMark, S\bigl[F''\;|\;\vec{v}\;|\; B, F_0\bigr]^{\vec{v}|d}_{\From P\mathbin{.}R}\bigr\}^d
\Rightarrow
  \bigl\{F\circ F_0, F', F_\FrMark, S\bigl[F''\circ F_0\;|\;\vec{v}\;|\; B\bigr]^{\vec{v}|d}_{\From P\mathbin{.}R}\bigr\}^d,\nonumber\\
\lambda^\Fld_{\From;\Ok}:&\ \bigl\{F, F', F_\FrMark, S\bigl[F''\;|\;\vec{v}\;|\; B, F_0\bigr]^{\vec{v}|d}_{\From P\mathbin{.}R}\Ok\bigr\}^d\Rightarrow
  \bigl\{F, F', F_\FrMark, S \bigl[F''\;|\;\vec{v}\;|\;\ttrue\bigr]^{\vec{v}|d}_{\From P\mathbin{.}R}\bigr\}^d.
\end{align}
 We keep applying $\FromUnfl$ until we cannot match
$\sigma(P_!\circ P_?\circ P_0)$ with iterator state.
Then we replace the iterator frame with $\delta(B)$ where $B$ is the iteration status,
$\delta(\ttrue)=\Ok$, and $\delta(\ffalse)=\Empty$. To prevent
premature application, rule schema $\FromEnd$ is conditional, where the condition uses
functions $\mu_{P, \vec{v}}:\Pattern\;s_1\ldots s_n\rightarrow\Succ$ defined for each
$\vec{v}\subseteq\Var(Q)$ and
pattern $P$ occurring in $Q$ with the single equation
$ \mu_{P,\vec{v}}(F\circ P_!\circ P_?\circ P_0,\vec{v})=\yes$ (cf. Equation~\eqref{EndFuncDefEq}):
\begin{equation}
  \FromEnd:\bigl\{F, F', F_\FrMark, S\bigl[F''\;|\;\vec{v}\;|\;B\bigr]^{\vec{v}|d}_{\From P\mathbin{.}R}\bigr\}^d
\Rightarrow\{F,F', F_\FrMark, S\delta(B)\}^d\ \text{if}\
(\mu_{P,\vec{v}}(F'', \vec{v})=\yes)=\ffalse.
\end{equation}

The last two rules  reduce the $\StateD$-terms with an empty stack or
stack containing only the $\Ok$ constant into a term constructed with $\New$ or $\Fail$, respectively:
\begin{equation}
\lambda_{\text{dml}}^\Ok:\{F, F',F_\FrMark, \Ok\}\Rightarrow\New(F\circ F', F_\FrMark),\quad
\lambda_{\text{dml}}^\Empty:\{F, F',F_\FrMark, \Empty\}\Rightarrow\Fail(F\circ F',F_\FrMark).
\end{equation}

\begin{theorem}
$\cR_{\Sigma, \cD}^\dml(Q)$ is a terminating rewriting system.
\end{theorem}
\begin{proof}
The proof is similar to the proof  of Theorem~\ref{CondTerminating}. Only subqueries
of the form $\From P\mathbin{.}R$, where $P$ is semiterminating but not terminating
(i.e., $P_0\neq\Empty$, but $P_?=\Empty$) require special care.
In this case, the signature of $\From\_.\_$ forces $R$ to be
 success assured, i.e., $R$'s evaluation always succeeds,
and hence removed facts matching $P_0$ are never restored.
 This ensures termination of  $\From P\mathbin{.}R$'s execution.
\end{proof}

The following useful observation can be trivially verified by examining the rule schemas:
\begin{lemma}
\label{SubLemmaDML}
Let $R$ be a DML subquery of $Q$.
Then, for all multisets of facts $F$, $F'$, $G$, $G'$, multisets $F_\FrMark$, $G_{\FrMark}$ of fresh facts, stacks $S$,
lists of variables $\vec{v}=v_1,\ldots,v_n$ and
of values $\vec{a}=a_1,\ldots,a_n$,
\begin{enumerate}
\item  $\{F, F',F_\FrMark,S[\vec{a}]^{\vec{v}|d}_R\}^d
\rightarrow^*_{\cR_1}\{G, F'\circ G', G_\FrMark,S\Ok\}^d$
 iff
$\{F,\Empty, F_\FrMark,[]^{|d}_{\sigma(R)}\}^d
\rightarrow^*_{\cR_2}\New(G\circ G',G_\FrMark)$,
\item  $\{F, F',F_\FrMark,S[\vec{a}]^{\vec{v}|d}_R\}^d\rightarrow^*_{\cR_1}
\{F, F', G_\FrMark,S\}^d$
iff
$\{F,\Empty, F_\FrMark,[]^{|d}_{\sigma(R)}\}^d\rightarrow^*_{\cR_2}\Fail(F,G_\FrMark)$.
\end{enumerate}
where  $\sigma:=\{\vec{a}/\vec{v}\}$, $\cR_1:=\cR_{\Sigma, \cD}^\dml(Q)$,
and $\cR_2:=\cR_{\Sigma, \cD}^\dml(\sigma(R))$.
\end{lemma}

As in the case of pure queries, we consider two DML queries equivalent if and only if  they can match
their results. We should not, however, distinguish results differing only by the
choice of fresh values.
Let $\cS_\FrMark$ denote the set of nominal sorts in $\Sigma_S$.
Let $\Nom(t)$ be the $\cS_\FrMark$-sorted set of nominal values contained in term $t$.
For any $\cS_\FrMark$-sorted bijection $\alpha:X\rightarrow Y$ between sets of nominal values
we denote by $\hat\alpha$ the natural extension of $\alpha$ to  terms $t$ such that
$\Nom(t)\subseteq X$. More precisely,
$\hat\alpha(x)=\alpha(x)$ if $x\in X$,  $\alpha(c)=c$ if
 $c$ is a  constant of non-nominal sort or a variable, and
$\hat\alpha(f(t_1,\ldots, t_n))= f(\hat\alpha(t_1),\ldots,\hat\alpha(t_n))$ if
$f(t_1,\ldots, t_n)$ is of non-nominal sort. With those notions we define equivalence on DML queries as follows:
\begin{definition}
Let $Q_1$ and $Q_2$ be two DML queries in $\QLang^\dml$.
We say that $Q_1$ is logically equivalent to $Q_2$, writing $Q_1\equiv Q_2$, if and only if,
for all ground substitutions $\sigma$ such that $\sigma(Q_1)$
and $\sigma(Q_2)$ are closed, all ground multisets of facts $F$, $G$, all
ground multisets of fresh facts $F_\FrMark$, $G_\FrMark$, and all $i\in\{1,2\}$
\begin{enumerate}
\item if $\Init_{\sigma(Q_i)}(F,F_\FrMark)\rightarrow^!\New(G,G_\FrMark)$
then there exist
multisets of fresh facts  $F^{'}_\FrMark$, $G^{'}_\FrMark$, multisets of facts $F'$, $G'$, and
an  $\cS_\FrMark$-sorted bijection $\alpha:\Nom(F)\cup\Nom(G)\rightarrow\Nom(F')\cup\Nom(G')$
such that  $F'=\hat\alpha(F)$,
 $G'=\hat\alpha(G)$, and
 $\Init_{\sigma(Q_{3-i})}(F',F^{'}_\FrMark)\rightarrow^!\New(G',G^{'}_\FrMark)$.
 \item if $\Init_{\sigma(Q_i)}(F,F_\FrMark)\rightarrow^!\Fail(F,G_\FrMark)$
then there exist
multisets of fresh facts  $F^{'}_\FrMark$, $G^{'}_\FrMark$, a multiset of facts $F'$, and
an  $\cS_\FrMark$-sorted bijection $\alpha:\Nom(F)\rightarrow\Nom(F')$
such that  $F'=\hat\alpha(F)$ and
 $\Init_{\sigma(Q_{3-i})}(F',F^{'}_\FrMark)\rightarrow^!\Fail(F',G^{'}_\FrMark)$.
\end{enumerate}
\end{definition}

The following result is an immediate consequence of Lemma~\ref{SubLemmaDML}:
\begin{lemma}
Logical equivalence on queries in $\QLang^\dml$ is an equivalence relation and a congruence, i.e., if
$\kappa$ is a position in a DML query $Q$  such that $Q|_\kappa$ is
a DML query, and $R\equiv Q|_\kappa$, then $Q\equiv Q[R]_\kappa$.
\end{lemma}

\begin{lemma}
\label{RenameDMLQuery}
For any closed DML query $Q$ in $\QLang^\dml$, and any renaming $\sigma$,
$Q\equiv\sigma(Q)$.
\end{lemma}

Lemma~\ref{LemmaClearSemDML}, proven similarly to Lemma~\ref{LemmaClearSem}, clarifies elements of rewriting semantics of DML queries in $\QLang^\dml$
(we leave $\From\_.\_$ evaluation where we are better off with
the rewriting definition).
\begin{lemma}
\label{LemmaClearSemDML}
For all ground multisets of facts $F$, $G$ and of fresh facts $F_\FrMark$, $G_\FrMark$, as well as
closed DML queries   $Q$, $Q_1$,
$Q_2$ and  all sentences $\phi$, the following holds:
\begin{enumerate}
\item If $\Init_Q(F, F_\FrMark)\rightarrow^!\Gamma$ then $\Gamma=\New(G, G_\FrMark)$
or $\Gamma=\Fail(F, G_\FrMark)$
 for some ground multiset of facts $G$ and ground multiset of fresh facts $G_\FrMark$.
 If $\Init_Q(F, F_\FrMark)\rightarrow^!\Fail(G, G_\FrMark)$
 then $F=G$.
\item If $\Init_{f}(F, F_\FrMark)\rightarrow^!\Gamma$ then $\Gamma=\New(F\circ f, F_\FrMark)$.
If $\Init_{\Empty}(F, F_\FrMark)\rightarrow^!\Gamma$ then $\Gamma=\Fail(F, F_\FrMark)$.
\eject
\item $\Init_{\phi\Rightarrow Q}(F, F_\FrMark)\rightarrow^!\New(G, G_\FrMark)$
iff
$\Init_\phi(F)\rightarrow^!\Res(\ttrue)$ and $\Init_Q(F, F_\FrMark)\rightarrow^!\New(G, G_\FrMark)$.
\item $\Init_{\phi\Rightarrow Q}(F, F_\FrMark)\rightarrow^!\Fail(F, G_\FrMark)$ iff
$\Init_\phi(F)\rightarrow^!\Res(\ffalse)$ or
$\Init_Q(F, F_\FrMark)\rightarrow^!\Fail(F, G_\FrMark)$.
\item $\Init_{Q_1\Con Q_2}(F, F_\FrMark)\rightarrow^!\Fail(F, G_\FrMark)$ iff
there exists a multiset of fresh facts $G^{'}_\FrMark$ such that
$\Init_{Q_1}(F, F_\FrMark)\rightarrow^!\Fail(F, G^{'}_\FrMark)$
and $\Init_{Q_2}(F, G^{'}_\FrMark)\rightarrow^!\Fail(F, G_\FrMark)$.
\item $\Init_{Q_1\Con Q_2}(F, F_\FrMark)\rightarrow^!\New(G, G_\FrMark)$
iff, for some ground multisets of facts $F'$, $F''$, $G'$, $G''$ such that
$G=G''\circ F'\circ F''$,  a ground multiset
of fresh facts $G_\FrMark^{'}$, and stacks $S,S'\in\{\Empty,\Ok\}$ where $S=\Ok$ or $S'=\Ok$, we have
$\Init_{Q_1}(F, F_\FrMark)\rightarrow^*\{G',F',G_\FrMark^{'},S\}^d$ and
$\Init_{Q_2}(G', G^{'}_\FrMark)\rightarrow^*\{G'',F'',G_\FrMark,S'\}^d$.
\end{enumerate}
\end{lemma}

\begin{lemma}
\label{EquivLemmaDML}
The following logical equivalences hold between queries in $\QLang^\dml$:
\begin{gather*}
\Empty\Con Q\equiv Q,\quad
Q\Con\Empty\equiv Q,\quad
Q_1\Con(Q_2\Con Q_3)\equiv (Q_1\Con Q_2)\Con Q_3,\\
\False\Rightarrow R\equiv \Empty,\quad
\neg\False\Rightarrow R\equiv R,\quad
\exists P\mathbin{.}\Empty\equiv \Empty
\end{gather*}
\end{lemma}

Lemma~\ref{EquivLemmaDML} is very similar to Lemma~\ref{EquivLemmaQueries}, except that
$\_\Con\_$ is not commutative.
$Q_1\Con Q_2$ may be non-equivalent with $Q_2\Con Q_1$ if, say, $Q_1$ deletes a fact
which is referred to in some pattern in $Q_2$.

\medskip
We need to generalize the notion of confluence, lest
rewriting paths leading to terms differing only by distinct choices of fresh values (as in the next example) are to be considered non-convergent.
\begin{example}
\label{ExampleTrivialConfl}
Let $\sort{I}$ be a nominal sort. Let $r:\Nat\rightarrow\Fact$, $s:\sort{I}\;\Nat\rightarrow\Fact$.
Consider DML query
$Q:=\From \Fresh{C_{\sort{I}}(x)}\circ\Ret{r(y)}\mathbin{.}s(x,y)$, and let $F:=r(1)\circ r(2)$.
Then
\begin{multline*}
\Init_Q\bigl(F,C_{\sort{I}}(\imath^{\sort{I}}_0)\bigr)
\rightarrow^*
\bigl\{F,\Empty, C_{\sort{I}}(\imath^{\sort{I}}_0),[F||\ffalse]^{|d}_Q\bigr\}^d
\xrightarrow{\FromUnfl}
\bigl\{F,\Empty, C_{\sort{I}}(\imath^{\sort{I}}_1),[r(2)||\ffalse,\Empty]^{|d}_Q
[\imath^{\sort{I}}_0,1]^{x,y|d}_{s(x,y)}\bigr\}^d\\
\rightarrow^*
\bigl\{F,s(\imath^{\sort{I}}_0,1), C_{\sort{I}}(\imath^{\sort{I}}_1),[r(2)||\ttrue]^{|d}_Q\bigr\}^d
\xrightarrow{\FromUnfl}
\bigl\{F,S(\imath^{\sort{I}}_0,1), C_{\sort{I}}(\imath^{\sort{I}}_2),[\Empty||\ttrue]^{|d}_Q
[\imath^{\sort{I}}_1,2]^{x,y|d}_{s(x,y)}\bigr\}^d\\
\rightarrow^*
\New\bigl(F\circ s(\imath^{\sort{I}}_0,1)\circ s(\imath^{\sort{I}}_1,2), C_{\sort{I}}(
\imath^{\sort{I}}_2)\bigr),
\end{multline*}
and, if we match $r(1)$ and $r(2)$ in reverse order in applications of $\FromUnfl$ rule, then
\begin{equation*}
\Init_Q\bigl(F,C_{\sort{I}}(\imath^{\sort{I}}_0)\bigr)
\rightarrow^*
\New\bigl(F\circ s(\imath^{\sort{I}}_0,2)\circ s(\imath^{\sort{I}}_1,1), C_{\sort{I}}(
\imath^{\sort{I}}_2)\bigr).
\end{equation*}
\end{example}

Here we define an equivalence relation  on terms of sort $\StateD$ which is a bisimulation:
\begin{definition}
We say that term $t_1$ is nominally equivalent to term $t_2$, in which case we write
$t_1\equiv_\FrMark t_2$,
if and only if there exists a bijection $\alpha:\Nom(t_1)\rightarrow\Nom(t_2)$
such that $\hat\alpha(t_1)=t_2$.
\end{definition}
\begin{lemma}
\label{NominEquivLem}
Nominal equivalence is an equivalence relation.  When restricted
to $\StateD$ terms satisfying the freshness condition (Definition~\ref{FreshCondRem}),  it is also a bisimulation
on $\cR_{\Sigma, \cD}^\dml(Q)$, for all $Q$.
\end{lemma}
We leave an easy proof of Lemma~\ref{NominEquivLem} to the reader.
The restriction to terms satisfying  the freshness condition is necessary for nominal equivalence
being a bisimulation, as demonstrated below:
\begin{example}
Let $\sort{I}$ be a nominal sort. Let $r:\sort{I}\rightarrow\Fact$, $s:\sort{I}\;\sort{I}\rightarrow\Fact$.
Define $F:=r(\imath^{\sort{I}}_1)\circ r(\imath^{\sort{I}}_2)$. Consider a DML query
$Q:=\From\Fresh{C_{\sort{I}}(x)}\circ\Ret{r(y)}\mathbin{.}(\{x=y\}\Rightarrow s(x,y)).$
Let $t_1:=\Init_{Q}(F, C_{\sort{I}}(\imath^{\sort{I}}_0))$
and $t_2:=\Init_{Q}(F, C_{\sort{I}}(\imath^{\sort{I}}_3))$.
Term $t_1$ does not satisfy the freshness condition (Definition~\ref{FreshCondRem}). It is
immediate that $t_1\equiv_\FrMark t_2$,
 $t_1\rightarrow^!t_3$, where
 $t_3:=\New(s(\imath^{\sort{I}}_1,1)\circ F,C_{\sort{I}}(\imath^{\sort{I}}_2))$,
 but the only normal form of  $t_2$ is $t_4:=\Fail(F, C_{\sort{I}}(\imath^{\sort{I}}_5))$ and $t_3\not\equiv_\FrMark t_4$.
\end{example}

We now define a
class of queries $Q$ for which $\cR_{\Sigma, \cD}^\dml(Q)$
is confluent modulo nominal equivalence.

\begin{definition}
Let $Q$ be a DML query. We say that $Q$ has {\em no deletion conflicts}
if and only if  for each DML subquery $\From P\mathbin{.}R$ of $Q$,
 and any subterm $f:\Fact$ of $R$ (resp. $P$) occurring inside $\Del{\_}$,
  $P$ (resp. $R$)
has no subterm $f'$ occurring inside
$\Del{\_}$, $\Keep{\_}$ or $\Ret{\_}$ such that $f$ and $f'$ are unifiable.
\end{definition}

\begin{definition}
Let $Q$ be a DML query in $\QLang^\dml$. $Q$ is called {\em deterministic} if it has no deletion conflicts
and
all quantification patterns in $Q$
(including those inside subterms which are conditions) contain
only single facts with unique matching property (but may contain any number of fresh facts).
\end{definition}

In Example~\ref{DeterministicNonconfluentEx} we shown why multiple facts in patterns lead to non-confluence, and as a result, to non-deterministic evaluation of conditions (and queries). However, we have not previously considered deletion conflicts (as they are specific to DML queries). The following example shows why deletion conflicts can prevent confluence:
\begin{example}
Consider the following DML query:
\begin{equation}
Q:=\From \Del{f(x)}\mathbin{.}\bigl(\Ok\Con (\From \Ret{f(1)}\mathbin{.} h(x))\bigr).
\end{equation}
Observe that all quantification patterns in $Q$ consist of single facts, but $Q$ does have deletion conflicts (facts in both patterns are unifiable, and one of them is $\Del{\_}$-pattern which has the second one in its scope). Since the first pattern ($\Del{f(x)}$) is semi-terminating but not terminating, to ensure that the DML query
in its scope is success assured it is of the form $\Ok\Con\_$.

Now, let $F:=f(1)\circ f(2)$. It is easy to see that executing $Q$ against $F$ removes from $F$ both $f$-facts and either adds a single fact $h(2)$ or nothing depending on whether
pattern $\Del{f(x)}$ first matches $f(2)$ (which makes it possible for the subquery $\From \Ret{f(1)}\mathbin{.} h(x)$ to succeed then and return $h(2)$) or $f(1)$ (which causes all executions of subquery $\From \Ret{f(1)}\mathbin{.} h(x)$ to fail).
\end{example}

The following theorem states that while evaluation of a deterministic DML query is not itself deterministic, but its results are.

\begin{theorem}
\label{ConfluentPartDML}
Let $Q$ be a deterministic query in $\QLang^\dml$. Then
$\cR_{\Sigma,\cD}^\dml(Q)$ is confluent up to a nominal equivalence. In particular, given  ground multisets of facts $F$,
and  of fresh facts $F_\FrMark$,
there is a unique (up to nominal equivalence) term $t$ of the form
$\New(G, G_\FrMark)$ or $\Fail(F, G_\FrMark)$ such that
$\Init_Q(F, F_\FrMark)\rightarrow^!t$.
\end{theorem}
\begin{proof}
The only significant difference between the proof of this theorem and Theorem~\ref{ConfluentPartQuery}
is the presence of deletions and fresh facts. The non-confluence introduced by
fresh facts can be absorbed with nominal equivalence.
Since $Q$ has no deletion conflicts, when a
DML subquery $\From P\mathbin{.}R$ is executed, neither
deletion of facts through $P$ influences execution of $R$ nor execution of $R$
decreases the pool of facts available for matching with $P$.
Moreover, if $P$ contains $\Del{f}$
for some fact $f$, then $f$ is the only fact in $P$, hence $R$ cannot fail,
facts deleted through $P$ are never returned, and $\Del{f}$ behaves like $\Ret{f}$.
\end{proof}

Let us finish this section with the following remark about expressibility of $\QLang^\dml$:
\begin{remark}
A typical formalization  of database updates is to use  pairs of queries
which define facts to be, respectively,
 deleted from, and added to the current database. This approach can be emulated in $\QLang^\dml$,
 using multiple DML queries executed in a sequence. First, let us extend the signature $\Sigma$ with
 function symbols $f^d$ and $f^a$ for each fact constructor $f$.
 Let $Q_d$ and $Q_a$ be queries in $\QLang^\query$ which return sets of facts to be
 deleted and added, respectively, to the database. We assume that $Q_d$ and $Q_a$ contain
 no subterms of the form $f^d(\vec{t})$ or $f^a(\vec{t})$.
 Let $\hat{Q}_d$ and $\hat{Q}_a$ be the same as $Q_d$ and $Q_a$, respectively, except that
 all subqueries $f(\vec{t})$ of sort $\Fact$ are replaced, respectively, with
 $f^d(\vec{t})$ and $f^a(\vec{t})$. Then to update the database
 we execute  the following DML queries (in this order):
 \begin{gather*}
 \hat{Q}_d,\quad\hat{Q}_a,\quad
 \From\Ret{f_1^d(\vec{v}^1)}\mathbin{.}\From\Del{f_1(\vec{v}^1)}\mathbin{.}\Ok,
 \ldots,
 \From\Ret{f_m^d(\vec{v}^m)}\mathbin{.}\From\Del{f_m(\vec{v}^m)}\mathbin{.}\Ok,
 \\
 \From\Del{f_1^d(\vec{v}^1)}\mathbin{.}\Ok,\ldots,
  \From\Del{f_m^d(\vec{v}^m)}\mathbin{.}\Ok,\quad
  \From\Del{f_1^a(\vec{v}^1)}\mathbin{.}f_1(\vec{v}^1),\ldots,
   \From\Del{f_m^a(\vec{v}^m)}\mathbin{.}f_m(\vec{v}^m),
 \end{gather*}
 where $f_1,\ldots, f_m$ are fact constructors occurring in $Q_d$ and $Q_a$.
 Thus, because of Theorem~\ref{RelAlgTheorem} we can express any {\em relational}
 database update using multiple DML queries in $\QLang^\dml$.
 \end{remark}

\section{Reachability analysis of data-centric business processes}
\label{ExampleActionsSection}

In this section we demonstrate the application of $\QLang^\dml$ in  specification and analysis
of data-centric business processes. First, we describe a general
reachability and simulation framework, and then devote the rest of the section to an extended example
specification.

So far, we have specified the rules for execution of a single DML expression. A business process, in general, executes a sequence of DML expressions according to some orchestration rules. A simple example of such rules which we use here, appropriate for a data driven process, is that if a DML expression can be executed successfully then it can be chosen
(non-deterministically) as the next command to be executed.
Usually  (c.f., \cite{hull2011business}) such data modifying operations are launched in response to some events, such as user actions which also provide input parameters for the command. In turn, their execution may trigger further events. Here, taking inspiration from  \cite{abiteboul1998relational,abiteboul2000relational}, we interpret some of the facts as triggering events, user input, and output events (we describe this in more detail later as a part of the example). This simplifies the simulation.

Thus, we specify a business process simply as a finite set $\Gamma$ of DML expressions in $\QLang^\dml$.
During simulation, at each ``business step'' a DML expression is chosen non-deterministically and is executed. Unsuccessful execution simply leaves the database of facts unchanged. Alternatively, during reachability analysis which performs a breadth-first search through all possible evolutions of the process it is more efficient to make the system stuck on unsuccessful step. This prunes spurious branches in search tree.

%\medskip
More precisely, a set of DML expressions
$\Gamma$ determines a rewriting system $\cR_{\Sigma, \cD}(\Gamma)$ defined to be the union of $\cR_{\Sigma, \cD}^\dml(Q)$'s, $Q\in\Gamma$,
 augmented with two rule schemas
\begin{equation}
\DmlNew_Q:\New(F, F_\FrMark)\Rightarrow\Init_Q(F, F_\FrMark),\quad
\DmlFail_Q:\Fail(F, F_\FrMark)\Rightarrow\Init_Q(F, F_\FrMark),
\end{equation}
for all $Q\in\Gamma$.

\medskip
Rule schema $\DmlNew_Q$ chooses non-deterministically a new DML query for execution and rewrites into an initial state for this query  if the execution of the previous one was successful.
Similarily,  $\DmlFail_Q$ chooses a new DML query if the execution of the previous one failed.
It is important to know that the failure of a randomly chosen DML expression does not usually mean that
the business process execution is faulty: Instead, it may simply mean that the given DML expression was
not applicable at the moment. Since here the only way to know if the DML expression is applicable is to
run it, the rule $\DmlFail_Q$ is necessary lest the simulation stops prematurely. On the other hand, unsuccessful executions do not change database state, and thus are spurious, adding no useful information. This is why, when doing reachability analysis which explores using breadth-first search all possible paths of execution (in contrast to simulations, where each simulation travels just a single execution path) it is better to drop the rule $\DmlFail_Q$.

 It is assumed that all DML queries $Q\in\Gamma$ are
such that a successful execution of Q consumes and emits a special fact $\Token$ (of sort $\Fact$) called a token.
The token does not denote any real data, but rather facilitates a non-deterministic choice of user input. Say, if in the database
there were facts $f(a_1), f(a_2), \ldots$, where $a_1, a_2, \ldots$ are possible user inputs for some business step, then we can simulate user choice and execution of further action $D(x)$ (based on this
choice and expressed as DML query with free variable $x$ storing user's decision) by using the DML
expression of the form $\From \Del{\Token}\circ\Ret{f(x)}\mathbin{.}D(x)$.
If the query wouldn't match and remove the token
then the action $D(x)$ would be executed for every possible user input. In the example
described in this section instead of a constant token we  use tokens $\Token : Nat \rightarrow \Fact$
parametrized by a natural number.
All DML queries consume a token with a non-zero parameter and emit a token with a parameter decreased by
one. This permits limiting the number of ``large business steps'', i.e., executions of DML queries, in
the simulation or search procedure. Rewriting systems such as Maude permit limiting rewriting steps in the
search procedure. However, execution of each DML query can take many rewriting steps, the number of which
is not easy to estimate. Thus, it is not trivial to pass from the number of rewriting steps to the number of
business steps (which are more natural in this context).

Given an initial database $F$ we start
 reachability analysis from term
 $\New(F\circ\Token(k), C_{s_1}(\imath_{m_1}^{s_1})\circ\cdots C_{s_n}(\imath_{m_n}^{s_n}))$,
 where $k$ is the maximal depth of search expressed in the number of business steps,  $s_1,\ldots,s_n$ are nominal sorts for which we need fresh values, and
 $m_1,\ldots,m_n$ are  such that values $\imath_{p_i}^{s_i}$
 for $p_i\geq m_i$,
 $i\in\{1,\ldots,n\}$, do not occur in $F$. In case of reachability analysis we search for the term of the form $\New(F, F_\FrMark)$ where the database $F$ satisfies some condition (either desirable or undesirable one).

\subsection{Example specification}

We borrow an example from \cite[Appendix~C]{abdulla2016recency}
to demonstrate specification of a business process as a set of DML expressions in $\QLang^\dml$.
The example concerns the process of selecting and advertising restaurant offers of dinners by employees of
mediating agency, and managing corresponding bookings. The lifecycles of two key artifact types
 --- {\em Offer} and {\em Booking} --- are presented as finite state machines in
Figure~\ref{LifecyclesBooking}. Each agent publishes exactly one restaurant offer ---
either the new one which just came or the one which was previously put on hold.
The published offer
is in the state $\Const{available}$.
Agent puts the offer he currently publishes on hold (state $\Const{onHold}$) when picking up another
offer for publication. Dashed  arrows in Figure~\ref{LifecyclesBooking} indicate that entering a given state by an artifact may trigger state change in another artifact, e.g., there is a dashed arrow between the $\Const{available}$ state and the
anonymous transition into $\Const{onHold}$ state (in a distinct artifact of type {\em Offer}). Available
 offer may get closed (state $\Const{closed}$, or be picked up by a customer (transition $\Const{newBooking}$
 to state $\Const{beingBooked}$).
The latter triggers creation of a new {\em Booking} artifact.
Booking starts with a preliminary phase of {\em drafting} (state $\Const{drafting}$) in which the customer
chooses dinner hosts (transition $\Const{addHosts}$). After draft submission
(which changes the state to $\Const{submitted}$)  the agent computes price for the offer
(transition $\Const{determineProposal}$ to state $\Const{finalized}$) and
the customer decides to either accept or reject the proposal transitioning, respectively,
to the $\Const{accepted}$ or $\Const{canceled}$ state. The acceptance may in some cases go through $\Const{toBeValidated}$ state when additional validation is necessary.

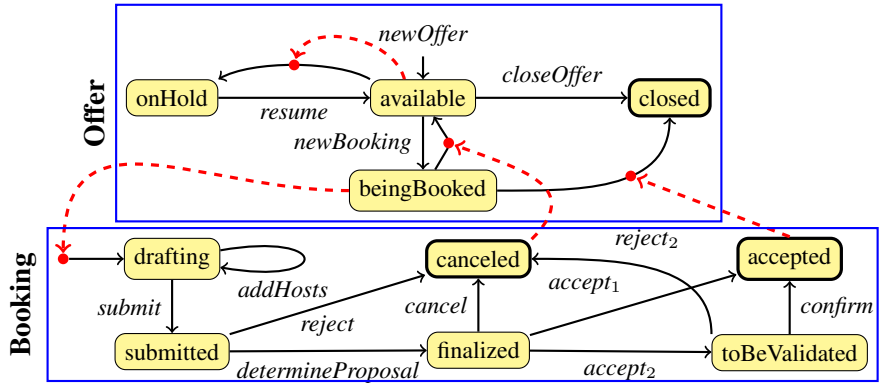
\begin{figure}[h]
\vspace*{3mm}
\centering
\begin{tikzpicture}[state/.style={draw=black, fill=yellow!50, rounded corners, font=\footnotesize}, action/.style={font=\footnotesize\it}, trigger/.style={fill=red, circle, inner sep=0.5mm}, every fit/.style={draw=blue, rectangle, thick}]
\draw (0,0) node(a1)[state]{onHold};
\draw (a1.east) node(a2)[state, anchor=west, xshift=2cm]{available};
\draw (a2.east) node(a3)[state, anchor=west, xshift=2cm, very thick]{closed};
\draw (a2.south) node(a4)[state, anchor=north, yshift=-0.7cm]{beingBooked};
\draw (a2.north) node(a5)[anchor=south, yshift=0.3cm, action]{newOffer};
\draw[->, thick] (a1) -- (a2) node[pos=0.5, anchor=north, action]{resume};
\draw[->, thick] (a2) -- (a3) node[pos=0.5, anchor=south, action]{closeOffer};
\draw[->, thick] (a2) -- (a4) node[pos=0.5, anchor=east, action]{newBooking};
\draw[->, thick] (a5) -- (a2);
\draw[->, thick] (a4.east) .. controls +(0:1cm) and +(270:1cm) .. (a3.south) node(d1)[pos=0.6, trigger]{};
\draw[->, thick] (a4.60) .. controls +(60:0.5cm) and +(300:0.5cm) .. (a2.300) node(d2)[pos=0.5, trigger]{};
\draw[->, thick] (a2.north west) .. controls +(150:0.5cm) and +(30:0.5cm) .. (a1.north east) node(d3)[pos=0.5, trigger]{};
\draw[->, dashed, red, very thick] (a2.135) .. controls +(105:0.6cm) and +(75:0.6cm) .. (d3);
\node[fit=(a1) (a3) (a4) (a5)] (f1) {};
\draw (f1.west) node[anchor=south, rotate=90]{\bf Offer};
\draw (a1.south) node(b1)[state, anchor=north, yshift=-1.6cm]{drafting};
\draw (b1.west) node(d4)[trigger, xshift=-0.8cm]{};
\draw[->, thick] (d4) -- (b1);
\draw[->, dashed, red, very thick] (a4.west) .. controls +(180:2cm) and +(90:2cm) .. (d4.north);
\draw[->, thick] (b1.10) .. controls +(10:1.5cm) and +(350:1.5cm) .. (b1.350) node(x1)[action, pos=0.75, anchor=north](k1){addHosts};
\draw (b1.east) node(b2)[state, anchor=west, xshift=2.7cm, very thick]{canceled};
\draw (b2.east) node(b3)[state, anchor=west, xshift=2.7cm, very thick]{accepted};
\draw[->, dashed, red, very thick] (b2.north east) .. controls +(45:1cm) and +(340:1cm) .. (d2);
\draw[->, dashed, red, very thick] (b3.north) -- (d1);
\draw (b1.south) node(b4)[state, anchor=north, yshift=-0.7cm]{submitted};
\draw (b2.south) node(b5)[state, anchor=north, yshift=-0.7cm]{finalized};
\draw (b3.south) node(b6)[state, anchor=north, yshift=-0.7cm]{toBeValidated};
\draw[->, thick] (b1) -- (b4) node[action, anchor=east, pos=0.5]{submit};
\draw[->, thick] (b5) -- (b2) node[action, anchor=east, pos=0.5]{cancel};
\draw[->, thick] (b6) -- (b3) node[action, anchor=west, pos=0.5]{confirm};
\draw[->, thick] (b4) -- (b2) node[action, anchor=north, pos=0.5]{reject};
\draw[->, thick] (b4) -- (b5) node[action, anchor=north, pos=0.5]{determineProposal};
\draw[->, thick] (b5) -- (b6) node[action, anchor=north, pos=0.5]{accept$_2$};
\draw[->, thick] (b5) -- (b3) node[action, anchor=south east, pos=0.5]{accept$_1$};
\draw[->, thick] (b6.north west) .. controls +(90:1cm) and +(0:1cm) .. (b2.east) node[action, anchor=south west, pos=0.7]{reject$_2$};
\node[fit=(b1) (d4) (b6)] (f2) {};
\draw (f2.west) node[anchor=south, rotate=90]{\bf Booking};
\end{tikzpicture}\vspace*{-1mm}
\caption{Lifecycles of {\em Offer} and {\em Booking} artifacts presented as finite state machines \cite[Figure~5]{abdulla2016recency} (see also \cite[Figure~5]{abdulla2016recencyext})}
\label{LifecyclesBooking}\vspace*{-2mm}
\end{figure}

\medskip
States of offers and bookings are constants of sorts $\sort{OState}$ and $\sort{BState}$, respectively,  named
like in Figure~\ref{LifecyclesBooking}.
We use the following nominal sorts for identifiers:
$\sort{Rest}$ for restaurants; $\sort{Person}$ for  customers, agents and hosts;
$\sort{Offer}$, $\sort{Book}$ and $\sort{Url}$
for  offers, bookings, and url's of finalized proposals, respectively.
For facts we use the following constructors:
\begin{gather*}
\Rel{Rest} : \sort{Rest}\rightarrow\Fact,\quad \Rel{Agent},\Rel{Cust} : \sort{Person}\rightarrow\Fact,\\
\Rel{Offer} : \sort{Offer}\;\sort{OState}\;\sort{Rest}\;\sort{Person}\rightarrow\Fact,\quad
\Rel{Book} : \sort{Book}\;\sort{BState}\;\sort{Offer}\;\sort{Person}\rightarrow\Fact,\\
%\Rel{Ost} : \sort{Offer}\;\sort{OState}\rightarrow\Fact,\quad
%\Rel{Bst} : \sort{Book}\;\sort{BState}\rightarrow\Fact,\quad
\Rel{Host} : \sort{Book}\;\sort{Person}\rightarrow\Fact,\quad
\Rel{Prop} : \sort{Book}\;\sort{Url}\rightarrow\Fact.
\end{gather*}
where facts $\Rel{Rest}(r)$, $\Rel{Agent}(a)$, $\Rel{Cust}(c)$ indicate that
$r$, $a$, and $c$ are, respectively, identifiers of a registered restaurant, agent, and
customer.
$\Rel{Offer}(o, s, r, a)$ means that an offer $o$ in a state $s$ for a restaurant $r$
is managed by an agent $a$.
$\Rel{Book}(b, s, o, c)$ means that booking $b$ in a state $s$ for a customer $c$ is
 related to an offer $o$.
A fact $\Rel{Host}(b,p)$ indicates that a person $p$
 is included as a host for booking $b$.
 Finally,  $\Rel{Prop}(b, u)$ indicates that  finalized proposal for booking $b$,
 with details and prices, is available at the url $u$.

\medskip
We now specify selected transitions from Figure~\ref{LifecyclesBooking} in detail.
Transition $\Const{newOffer}$ responsible for creation of new offers
is implemented with the following DML query:
\begin{multline*}
\From \Del{\Token(s(n))}\circ\Fresh{C_{\sort{Offer}}(o)}\circ \Ret{\Rel{Agent}(a)}
\circ\Keep{\Rel{Rest}(r)}\mathbin{.}
\bigl(\forall\;\Ret{\Rel{Offer}(o', \Const{beingBooked}, r', a)}\;.\;\False\bigr)\\
\Rightarrow\bigl(
O(o, \Const{available}, r, a)\Con
\bigl(\From \Del{O(o', \Const{available}, r', a)}\mathbin{.}
O(o', \Const{onHold}, r', a)
\bigr)\Con \Token(n)\bigr).
\end{multline*}
Above, $s:\Nat\rightarrow\Nat$ denotes the successor function.
We assume that each DML query is executed against a database in which there is exactly one token
matching $\Token(s(n))$, i.e.,  a token holding a number greater than zero.
Thus, in the above, we first choose a {\em single} fresh offer identifier, and, non-deterministically, a {\em single} registered
agent and a {\em single} restaurant. The token is marked with $\Del{\_}$, so it is removed from the database after matching (this guarantees that we choose no more than one agent, restaurant, and offer identifier).
The query emits back a token with a number decreased by 1 ensuring the possibility (if this number is greater than zero) of executing a next query.
The restaurant can be arbitrary (as long as it is registered in the system),
 however the agent must not manage an offer being booked,
as described by the deterministic condition
$$\forall\Ret{O(o', \Const{beingBooked}, r', a)}\;.\;\False$$ inside the above
DML query. If this condition is not satisfied, the quantifier step fails, the
token is returned to the database and a new matching is tried. Since $\Rel{Agent}(a)$
is marked by $\Ret{\_}$, we do not try the same agent again.
If the correct matching is found, a fact describing new offer
is added to the database. We also change the state of any available offer managed by
the agent of the new offer to $\Const{onHold}$.

\medskip
An offer which was put on hold, may be resumed
by any agent who is not managing an offer which is currently being booked.
Agent resuming an offer becomes the new manager of the offer:
\begin{multline*}
\From \Del{\Token(s(n))\circ \Rel{Offer}(o,\Const{onHold}, r, a)}\circ\Ret{\Rel{Agent}(a')}
\mathbin{.}\bigl(\forall\;\Ret{\Rel{Offer}(o', \Const{beingBooked}, r', a')}\;.\;\False\bigr)\\
\Rightarrow\bigl(\Token(n)\Con
\Rel{Offer}(o, \Const{available}, r, a')\\
\Con
\bigl(\From \Del{\Rel{Offer}(o', \Const{available}, r', a')}\mathbin{.}
\Rel{Offer}(o', \Const{onHold}, r', a')
\bigr)\bigr).
\end{multline*}

With the $\Const{newBooking}$ transition some offer $o$  changes state from $\Const{available}$  to $\Const{beingBooked}$).
It also triggers creation of a new booking (with a fresh identifier) in the
$\Const{drafting}$ state for the chosen offer $o$ on behalf of some registered customer:
\begin{multline*}
\From \Del{\Token(s(n))\circ \Rel{Offer}(o,\Const{available}, r, a)}\circ\Fresh{C_{\sort{Book}}(b)}
\circ\Keep{\Rel{Cust}(c)}\\
\mathbin{.}\bigl(
\Rel{Offer}(o,\Const{beingBooked}, r, a)\Con \Rel{Book}(b,\Const{drafting}, o, c)
\Con\Token(n)
\bigr).
\end{multline*}

The customer involved in booking can  add
dinner hosts one by one  (see transition $\Const{addHosts}$ in Figure~\ref{LifecyclesBooking}) as long as the booking is in the drafting stage.
The added host can be either fresh or be present in the database as a host
for another offer. We use separate DML queries for each of those cases.
The first case (of a fresh host) is trivial:
\begin{equation*}
\From \Del{\Token(s(n))}\circ\Fresh{C_{\sort{Person}}(h)}\circ
\Keep{\Rel{Book}(b,\Const{drafting},o,c)}
\mathbin{.}\bigl(\Token(n)\Con\Rel{Host}(b,h)\bigr).
\end{equation*}
In the second case we have to ensure that we are not adding the same person twice:
\begin{equation*}
\From \Del{\Token(s(n))}\circ\Keep{\Rel{Book}(b,\Const{drafting},o,c)}\circ\Ret{\Rel{Host}(b',h)}
\mathbin{.}\bigl((
\forall\;\Ret{\Rel{Host}(b,h)}\mathbin{.}\False)
\Rightarrow (\Token(n)\Con\Rel{Host}(b,h))\bigr).
\end{equation*}

The $\Const{submit}$ action simply changes the state of the booking from
$\Const{drafting}$ to $\Const{submitted}$. Then, if the customer's customized booking is infeasible, it can be rejected ($\Const{reject}$ transition in Figure~\ref{LifecyclesBooking}, the implementation of which we omit for brevity's sake).
Otherwise, the final proposal (which includes cost, etc.) to the customer who
owns the booking is created
(transition $\Const{determineProposal}$ in Figure~\ref{LifecyclesBooking}).
The preparation of the proposal is abstracted  as (1) creating the fresh url
to the proposal, and (2) removing information about hosts (which is available at the new url).
As before, we pick the booking non-deterministically using the token and appropriate
pattern:
\begin{multline*}
\From \Del{\Token(s(n))\circ \Rel{Book}(b,\Const{submitted},o,c)}\circ\Fresh{C_{\sort{Url}}(u)}
\mathbin{.}\\
\bigl(\Rel{Prop}(b,u)\Con
\Rel{Book}(b,\Const{finalized},o,c)\Con
\bigl(
\From \Del{\Rel{Host}(b,h)}\mathbin{.}\Ok
\bigr)\Con\Token(n)\bigr).
\end{multline*}

A finalized booking proposal  for a restaurant $r$ can be accepted either immediately (with
$\Const{accept}_1$)  or after an additional confirmation (with $\Const{accept}_2$).
The first case applies only to
golden customers of a given restaurant $r$, i.e.,  those who successfully booked
a dinner in $r$ at least $k$-times, for some fixed $k$.
Accepting a proposal changes the state of the offer to which
the booking belongs to $\Const{closed}$:
 \begin{multline*}
 \From \Del{\Token(s(n))\circ \Rel{Book}(b,\Const{finalized},o,c)\circ \Rel{Offer}(o, \Const{beingBooked}, r, a)}
 \circ\Ret{\Rel{Cust}(c)}
 \mathbin{.}\\
 \bigl(
 \exists\;[\Rel{Offer}(o_1,\Const{closed}, r, a_1)\circ B(b_1,\Const{accepted},o_1,c)\\
 \circ\cdots\circ
 \Rel{Offer}(o_k,\Const{closed}, r, a_k)\circ \Rel{Book}(b_k,\Const{accepted},o_k,c)]_?\mathbin{.}\True
 \bigr)\\
 \Rightarrow
 \bigl(\Rel{Book}(b,\Const{accepted},o,c)\Con \Rel{Offer}(o, \Const{closed}, r, a)\Con\Token(n)\bigr).
 \end{multline*}

 \begin{remark}
 In our earlier work \cite{Bartek2017} we have used the almost same example (with minor modifications) to illustrate an alternative formalism (c.f. Section~\ref{SectPrior} in the current paper) where queries, also implemented in the rewriting system, but using meta-level features, are deterministic. This makes them behave like a classical queries, but because of determinism it is not possible to simulate user input directly. Instead, a separate mechanism had to be introduced to simulate non-deterministic input choice. Secondly, DML expressions in
 \cite{Bartek2017}  did not add or, more importantly, delete facts from the database directly. Instead, they return pairs (which need to be specified in the DML expression itself) of (multi)sets of facts: those to be deleted and those to be added. However, in this particular example (and we believe it is typical) facts to be deleted are matched by parts of the
 patterns in the query. This (in \cite{Bartek2017}) led to code duplication, and suggested natural use of rewriting rules (which, of course, replaces matched subterms), and ultimately led to the formalism described in this paper.
 \end{remark}

\begin{remark}
We have implemented both the syntax and semantics of $\QLang^\Cond$ and $\QLang^\Dml$ in Maude
\cite{Maude2:03}. The implementation is available on the project's website \cite{ndqrl}.
To test the implementation we have used specification of the business process described above (also available from \cite{ndqrl}). The specification compiles into a Maude's system module which contains
definitions of 196 operators, 285 equations and 83 rewriting rules (in actual implementation we used equations in place of  deterministic rewriting rules).
\end{remark}

Let us now describe a simple example of a reachability analysis with the specification just described.
Let
\begin{equation*}
\Const{initDB}:=  \Rel{agent}(a_1) \circ \Rel{agent}(a_2) \circ \Rel{cust}(c_1)
\circ \Rel{cust}(c_2) \circ \Rel{rest}(r_1) \circ \Rel{rest}(r_2) \circ \Token(7)
\end{equation*}
be an initial database of facts. Let
\begin{equation*}
\Const{initDBN}:=
C(\imath^{\sort{Offer}}_0) \circ C(\imath^{\sort{Book}}_0) \circ
 C(\imath^{\sort{Url}}_0) \circ C(\imath^{\sort{Person}}_0)
\end{equation*}
be an initial multiset of fresh facts. Finally, let
\begin{equation*}
\Const{initState}:=\New(\Const{initDB}, \Const{initDBN})
\end{equation*}
be an initial state. Note that since the token is parametrized by 7, it follows that we can make at most
seven successful business steps from this state. We are interested in checking if we can reach from
 $\Const{initState}$ (in no more than 7 business steps) the state in which there exists a closed offer (and accepted associated booking)  from the database with no bookings and offers. Formally, we want to reach the state  matching
\begin{equation*}
\New(F\circ \Rel{Offer}(o_1, \Const{closed}, r_1, a_1) \circ
\Rel{Book}(b_1, \Const{accepted}, o_1, c_1), F_\FrMark).
\end{equation*}

Using our implementation \cite{ndqrl} we can easily check that a matching state is indeed reachable
in 6 business steps (Maude reported 525165 actual rewritings in 2528ms).

\section{Conclusion}

We have presented a multiset non-deterministic query and data manipulation language
$\QLang$ based on conditional term rewriting. The intended application of this language is in specification, simulation
and reachability analysis of data-centric business processes. However,  the remarkable features of $\QLang$, particularly non-determinism and non-standard approach to variable binding, make it  interesting on its own. We show that non-determinism of queries is useful for simulating user choices, but we also provide  easily identifiable syntactic restrictions which ensure uniqueness of query results.
Interestingly, this non-determinism leads to bisimulation-like definitions
of logical equivalence between formulas.
In the last section we demonstrated how sets of DML queries can be used to specify a business process and we provide a simple framework for simulation and testing.

$\QLang$ is a multiset query language. Most formal query languages, including
relational calculus and algebra, are based on sets. One under-appreciated fact is that
SQL is really a multiset query language, and for a very good reason --- removing
duplicates is expensive. While this was not our primary reason to use multisets,
 we believe that using multiset languages encourages query design which avoids
unnecessary expensive operations, and
takes the complexity of query execution into account better than set-based languages.

The fact that closed $\QLang$ formulas  are compiled to rewriting systems
permits their symbolic execution using narrowing \cite{meseguer2007symbolic}. We intend to explore this
possibility in future research. This is also one of the reasons why it was important to limit
the use of conditional rules as much as possible:  many implementations of narrowing
(see e.g., \cite{MaudeMan}) do not permit narrowing with conditions.

We have implemented $\QLang^\cond$, $\QLang^\dml$ and a specification framework extending the one described at the beginning of Section~\ref{ExampleActionsSection} in Maude \cite{MaudeMan}.
The implementation  is available from~\cite{ndqrl}. It differs in non-essential way from the one described in the present paper, but the code is extensively documented.

\subsection*{Acknowledgements}
The author is grateful to the anonymous reviewers for their helpful remarks

\end{document}